\documentclass[11pt]{article}
\usepackage{times}
\usepackage{amsmath}
\usepackage{amssymb}
\usepackage{verbatim}
\usepackage{graphicx}
\usepackage{epsfig}
\usepackage{algorithmic}
\usepackage{algorithm}
\usepackage{fullpage}
\usepackage{color}
\usepackage{graphicx}
\usepackage{epsfig}
\usepackage{amsthm}
\usepackage{latexsym}
\usepackage{amssymb}
\usepackage{amsmath}

\newcommand{\newfontobj}[2]{
  \newcommand{#1}[1]{
    \expandafter\def\csname##1\endcsname{{#2 ##1}}}}

\newfontobj{\class}{\rm} 

\class{PSPACE}	
\class{L}
\class{BPL}
\class{RL}
\class{NC}
\class{ZPL}
\class{NPSPACE}	
\class{ASPACE}	
\class{NL}
\class{EXP}
\class{NEXP}
\class{coNEXP}
\class{NSUBEXP}
\class{NE}
\class{E}
\class{AM}		
\class{MA}
\class{NP}
\class{DNP}
\class{UP}
\class{P}
\class{RP}
\class{BPP}
\class{ZPP}
\class{EXPSPACE}
\class{coNP}
\class{coRP}
\class{coAM}
\class{PH}
\class{IP}
\class{PCP}
\class{MIP}
\class{AC}
\class{TC}
\class{BP}

\DeclareMathOperator{\F}{\mathbb{F}}

\newtheorem{theorem}{Theorem}[section]
\newtheorem{proposition}[theorem]{Proposition}
\newtheorem{claim}{Claim}
\newtheorem{lemma}[theorem]{Lemma}

\newtheorem{definition}{Definition}[section] 

\newtheoremstyle{example}{\topsep}{\topsep}%
     {\normalfont \small}   
     {}    
     {\bfseries}     
     {}
     {\topsep}
     {}
\theoremstyle{example}

\newcommand{\argmax}{\mathop{\operatorname{argmax}}}

\newcommand{\Feas}{\mathcal{F}}
\newcommand{\OPT}{O^*}
\newcommand{\FOPT}{F^*}

\newcommand{\cost}{\mathsf{C}}
\newcommand{\profit}{\pi}
\newcommand{\incprofit}{\tilde{\profit}}
\newcommand{\val}{v}
\newcommand{\dcost}{\mathsf{c}}
\newcommand{\dcostinv}{\bar{\mathsf{s}}}

\newcommand{\const}[1][i]{k_{#1}}
\newcommand{\alg}{\mathcal{A}}

\newcommand{\eone}{\beta} 
\newcommand{\etwo}{\beta'}

\newcommand{\sdf}{h}
\newcommand{\fdf}{g}
\newcommand{\supp}{\operatorname{supp}}

\newcommand{\prefi}[1][i]{P_{(#1)}}

\newcommand{\I}{\mathbb{I}}
\newcommand{\Fr}{\mathbb{F}}

\newcommand{\deltac}{ {\rho^{\textrm{--}}}}
\newcommand{\deltab}{{\rho^+}}

\newcommand{\rhoc}{{ \rho^{ \textrm{--}}}}
\newcommand{\rhoci}[1][i]{{ \rho^{ \textrm{--}}_{#1}}}
\newcommand{\rhob}{{\rho^+}}

\newcommand{\preflarge}{P_{\rhoc}}
\newcommand{\prefsmall}{P_{\rhob}}
\newcommand{\guess}{\tau}
\newcommand{\prefguess}{P_{\guess}}

\newcommand{\densityv}{\rho}
\newcommand{\ideal}{\Pi}
\newcommand{\ebar}{e^{\prime}}
\newcommand{\invideal}{\Pi^{-1}}

\newcommand{\fdfmax}{G}

\newcommand{\FP}{\operatorname{P}^*}
\newcommand{\eonep}{\beta^{\prime \prime}}
\newcommand{\prefguessbar}{\bar{P}_\guess}
\newcommand{\prefsmallbar}{\bar{P}_\rhob}
\newcommand{\preflargebar}{\bar{P}_\rhoc}
\newcommand{\bP}{\bar{P}}

\newcommand{\lOneSize}{8}
\newcommand{\lOneMinK}{24}
\newcommand{\lOneTune}{.47}
\newcommand{\lOneRatio}{1/10}
\newcommand{\lOneRatioTable}{0.100}

\newcommand{\wSet}{I}
\newcommand{\gSet}{G}
\newcommand{\bSet}{B}

\newcommand{\decRate}{y}
\newcommand{\lOneApprox}{\beta(c)}

\newcommand{\wgt}{w}
\newcommand{\minWgt}{\underline{\wgt}}
\newcommand{\maxWgt}{\overline{\wgt}}

\newcommand{\floor}[1]{\left\lfloor{#1}\right\rfloor}
\newcommand{\abs}[1]{\left\lvert{#1}\right\rvert}

\newenvironment{proofsketch}{\noindent{\em Proof Sketch:}}{\qed}

\begin{document}

\title{Secretary Problems with Convex Costs}

\author{
Siddharth Barman \and Seeun Umboh \and Shuchi Chawla \and David Malec\\
~\\
University of Wisconsin--Madison\\
{\tt \{sid, seeun, shuchi, dmalec\}@cs.wisc.edu}
}
\date{}
\maketitle 


\begin{abstract}
We consider online resource allocation problems where given a set of
requests our goal is to select a subset that maximizes a value minus
cost type of objective function. Requests are presented online in
random order, and each request possesses an adversarial value and an
adversarial size. The online algorithm must make an irrevocable
accept/reject decision as soon as it sees each request. The ``profit''
of a set of accepted requests is its total value minus a convex cost
function of its total size. This problem falls within the framework of
secretary problems. Unlike previous work in that area, one of
the main challenges we face is that the objective function can be
positive or negative and we must guard against accepting requests that
look good early on but cause the solution to have an arbitrarily large
cost as more requests are accepted. This requires designing new techniques.


We study this problem under various feasibility constraints and
present online algorithms with competitive ratios only a constant
factor worse than those known in the absence of costs for the same
feasibility constraints. We also consider a multi-dimensional version
of the problem that generalizes multi-dimensional knapsack within a
secretary framework. In the absence of any feasibility constraints, we
present an $O(\ell)$ competitive algorithm where $\ell$ is the number
of dimensions; this matches within constant factors the best known
ratio for multi-dimensional knapsack secretary.
\end{abstract}


\setcounter{page}{1}

\section{Introduction}
\label{sec:intro}
We study online resource allocation problems under a natural profit objective: a single server accepts or rejects requests for service so as to maximize the total value of the accepted requests minus the cost imposed by them on the system. This model captures, for example, the optimization problem faced by a cloud computing service accepting jobs, a wireless access point accepting connections from mobile nodes, or an advertiser in a sponsored search auction deciding which keywords to bid on. In many of these settings, the server must make accept or reject decisions in an online fashion as soon as requests are received without knowledge of the quality future requests. We design online algorithms with the goal of achieving a small competitive ratio---ratio of the algorithm's performance to that of the best possible (offline optimal) solution.



A classical example of online decision making is the secretary problem. Here a company is interested in hiring a candidate for a single position; candidates arrive for interview in \emph{random order}, and the company must accept or reject each candidate following the interview. The goal is to select the best candidate as often as possible. What makes the problem challenging is that each interview merely reveals the rank of the candidate relative to the ones seen previously, but not the ones following. Nevertheless, Dynkin~\cite{Dynkin1963} showed that it is possible to succeed with constant probability using the following algorithm: unconditionally reject the first $1/e$ fraction of the candidates; then hire the next candidate that is better than all of the ones seen previously. Dynkin showed that as the number of candidates goes to infinity, this algorithm hires the best candidate with probability approaching $1/e$ and in fact this is the best possible.


More general resource allocation settings may allow picking multiple candidates subject to a certain feasibility constraint. We call such a problem a generalized secretary problem (GSP) and use $(\Phi,\Feas)$ to denote an instance of the problem. Here $\Feas$ denotes a feasibility constraint that the set of accepted requests must satisfy (e.g. the size of the set cannot exceed a given bound), and $\Phi$ denotes an objective function that we wish to maximize. As in the classical setting, we assume that requests arrive in random order; the feasibility constraint $\Feas$ is known in advance but the quality of each request, in particular its contribution to $\Phi$, is only revealed when the request arrives.
Recent work has explored variants of the GSP where $\Phi$ is the sum over the accepted requests of the ``value'' of each request. For such a sum-of-values objective, constant factor competitive ratios are known for various kinds of feasibility constraints including cardinality constraints~\cite{HKP04,Kleinberg2005multiple}, knapsack constraints~\cite{babaioff2007knapsack}, and certain matroid constraints~\cite{babaioff2007matroids}.

In many settings, the linear sum-of-values objective does not adequately capture the tradeoffs that the server faces in accepting or rejecting a request, and feasibility constraints provide only a rough approximation. Consider, for example, a wireless access point accepting connections. Each accepted request improves resource utilization and brings value to the access point. However as the number of accepted requests grows the access point performs greater multiplexing of the spectrum, and must use more and more transmitting power in order to maintain a reasonable connection bandwidth for each request. The power consumption and its associated cost are {\em non-linear} functions of the total load on the access point. This directly translates into a value minus cost type of objective function where the cost is an increasing function of the load or total size of all the requests accepted.

Our goal then is to accept a set $A$ out of a universe $U$ of requests such that the ``profit'' $\profit(A) = \val(A)-\cost(s(A))$ is maximized; here $\val(A)$ is the total value of all requests in $A$, $s(A)$ is the total size, and $\cost$ is a known increasing convex cost function\footnote{Convexity is crucial in obtaining any non-trivial competitive ratio---if the cost function were concave, the only solutions with a nonnegative objective function value may be to accept everything or nothing.}. 

Note that when the cost function takes on only the values $0$ and $\infty$ it captures a knapsack constraint, and therefore the problem $(\profit,2^U)$ (i.e. where the feasibility constraint is trivial) is a generalization of the knapsack secretary problem~\cite{babaioff2007knapsack}. We further consider objectives that generalize the $\ell$-dimensional knapsack secretary problem. Here, we are given $\ell$ different (known) convex cost functions $\cost_i$ for $1 \leq i \leq \ell$, and each request is endowed with $\ell$ sizes, one for each dimension. The profit of a set is given by $\profit(A) = \val(A)-\sum_{i=1}^\ell \cost_i(s_i(A))$ where $s_i(A)$ is the total size of the set in dimension $i$.

We consider the profit maximization problem under various feasibility constraints. For single-dimensional costs, we obtain online algorithms with competitive ratios within a constant factor of those achievable for a sum-of-values objective with the same feasibility constraints. For $\ell$-dimensional costs, in the absence of any constraints, we obtain an $O(\ell)$ competitive ratio. We remark that this is essentially the best approximation achievable even in the offline setting: Dean et al.~\cite{dean2005pip} show an $\Omega(\ell^{1-\epsilon})$ hardness for the simpler $\ell$-dimensional knapsack problem under a standard complexity-theoretic assumption. For the multi-dimensional problem with general feasibility constraints, our competitive ratios are worse by a factor of $O(\ell^5)$ over the corresponding versions without costs. Improving this factor is a possible avenue for future research.

We remark that the profit function $\profit$ is a submodular function. 
Recently several works ~\cite{FNS11, bateni2010submodular, gupta2010constrained} have looked at secretary problems with submodular objective functions and developed constant competitive algorithms.
However, all of these works make the crucial assumption that the objective is always nonnegative; it therefore does not capture $\profit$ as a special case. In particular, if $\Phi$ is a monotone increasing submodular function (that is, if adding more elements to the solution cannot decrease its objective value), then to obtain a good competitive ratio it suffices to show that the online solution captures a good fraction of the optimal solution. In the case of \cite{bateni2010submodular} and \cite{gupta2010constrained}, the objective function is not necessarily monotone. Nevertheless, nonnegativity implies that the universe of elements can be divided into two parts, over each of which the objective essentially behaves like a monotone submodular function in the sense that adding extra elements to a good subset of the optimal solution does not decrease its objective function value. In our setting, in contrast, adding elements with too large a size to the solution can cause the cost of the solution to become too large and therefore imply a negative profit, even if the rest of the elements are good in terms of their value-size tradeoff. As a consequence we can only guarantee good profit when no ``bad'' elements are added to the solution, and must ensure that this holds with constant probability. This necessitates designing new techniques.

\paragraph{Our techniques.} 

In the absence of feasibility constraints (see Section~\ref{sec:unconstrained}), we note that it is possible to classify elements as ``good'' or ``bad'' based on a threshold on their value to size ratio (a.k.a. density) such that any large enough subset of the good elements provides a good approximation to profit; the optimal threshold is defined according to the offline optimal fractional solution. Our algorithm learns an estimate of this threshold from the first few elements (that we call the sample) and accepts all the elements in the remaining stream that cross the threshold. Learning the threshold from the sample is challenging. First, following the intuition about avoiding all bad elements, our estimate must be conservative, i.e. exceed the true threshold, with constant probability. Second, the optimal threshold for the sample can differ significantly from the optimal threshold for the entire stream and is therefore not a good candidate for our estimate. Our key observation is that the optimal profit over the sample is a much better behaved random variable and is, in particular, sufficiently concentrated; we use this observation to carefully pick an estimate for the density threshold.


With general feasibility constraints, it is no longer sufficient to merely classify elements as good and bad: an arbitrary feasible subset of the good elements is not necessarily a good approximation. Instead, we decompose the profit function into two parts, each of which can be optimized by maximizing a certain sum-of-values function (see Section~\ref{sec:constrained}). This suggests a reduction from our problem to two different instances of the GSP with sum-of-values objectives. The catch is that the new objectives are not necessarily non-negative and so previous approaches for the GSP don't work directly. We show that if the decomposition of the profit function is done with respect to a good density threshold and an extra filtering step is applied to weed out bad elements, then the two new objectives on the remaining elements are always non-negative and admit good solutions. At this point we can employ previous work on GSP with a sum-of-values objective to obtain a good approximation to one or the other component of profit.
We note that while the exposition in Section~\ref{sec:constrained}
focuses on a matroid feasibility constraint, the results of that
section extend to any downwards-closed feasibility constraint that
admits good offline and online algorithms with a sum-of-values
objective\footnote{We obtain an $O(\alpha^4\beta)$ competitive
algorithm where $\alpha$ is the best offline approximation and $\beta$
is the best online competitive ratio for the sum-of-values objective.}.

In the multi-dimensional setting (discussed in Section~\ref{sec:multi}), elements have different sizes along different dimensions. Therefore, a single density does not capture the value-size tradeoff that an element offers. Instead we can decompose the value of an element into $\ell$ different values, one for each dimension, and define densities in each dimension accordingly. This decomposes the profit across dimensions as well. Then, at a loss of a factor of $\ell$, we can approximate the profit objective along the ``best'' dimension. The problem with this approach is that a solution that is good (or even best) in one dimension may in fact be terrible with respect to the overall profit, if its profit along other dimensions is negative. Surprisingly we show that it is possible to partition values across dimensions in such a way that there is a {\em single} ordering over elements in terms of their value-size tradeoff that is respected in each dimension; this allows us to prove that a solution that is good in one dimension is also good in other dimensions. We present an $O(\ell)$ competitive algorithm for the unconstrained setting based on this approach in Section~\ref{sec:multi}, and defer a discussion of the constrained setting to Section~\ref{app:multi-constrained}.

\paragraph{Related work.}
The classical secretary problem has been studied extensively; see \cite{Fer89, Fre83} and \cite{Sam91} for a survey. Recently a number of papers have explored variants of the GSP with a sum-of-values objective.  Hajiaghayi et al.~\cite{HKP04} considered the variant where up to $k$ secretaries can be selected (a.k.a. the $k$-secretary problem) in a game-theoretic setting and gave a strategyproof constant-competitive mechanism. Kleinberg~\cite{Kleinberg2005multiple} later showed an improved $1-O(1/\sqrt{k})$ competitive algorithm for the classical setting. Babaioff et al.~\cite{babaioff2007knapsack} generalized this to a setting where different candidates have different sizes and the total size of the selected set must be bounded by a given amount, and gave a constant factor approximation. In \cite{babaioff2007matroids} Babaioff et al. considered another generalization of the $k$-secretary problem to matroid feasibility constraints. A matroid is a set system over $U$ that is downwards closed (that is, subsets of feasible sets are feasible), and satisfies a certain exchange property (see \cite{Oxley1992} for a comprehensive treatment). They presented an $O(\log r)$ competitive algorithm, where $r$ is the rank of the matroid, or the size of a maximal feasible set. This was subsequently improved to a $O(\sqrt{\log r})$-competitive algorithm by Chakraborty and Lachish~\cite{chakrabortyimproved}. Several papers have improved upon the competitive ratio for special classes of matroids~\cite{babaioff2009secretary,dimitrov2008,korulapal2009}. Bateni et al.~\cite{bateni2010submodular} and Gupta et al.~\cite{gupta2010constrained} were the first to (independently) consider non-linear objectives in this context. They gave online algorithms for non-monotone nonnegative submodular objective functions with competitive ratios within constant factors of the ratios known for the sum-of-values objective under the same feasibility constraint. Other versions of the problem that have been studied recently include: settings where elements are drawn from known or unknown distributions but arrive in an adversarial order~\cite{CHMS10, kennedy1987prophet, samuel1984comparison}, versions where values are permuted randomly across elements of a non-symmetric set system~\cite{soto2010matroid}, and settings where the algorithm is allowed to reverse some of its decisions at a cost~\cite{babaioff2008selling, babaioff2009selling}.

\section{Notation and Preliminaries}
\label{sect:notation}
We consider instances of the generalized secretary problem represented
by the pair $(\profit,\Feas)$, and an implicit number $n$ of requests
or elements that arrive in an online fashion. $U$ denotes the universe
of elements. $\Feas\subset 2^U$ is a known downwards-closed feasibility
constraint. Our goal is to accept a subset of elements $A\subset U$
with $A\in\Feas$ such that the objective function $\profit(A)$ is
maximized. For a given set $T\subset U$, we use $\OPT(T)
= \argmax_{A \in \Feas\cap 2^T} \profit(A)$ to denote the optimal
solution over $T$; $\OPT$ is used as shorthand for $\OPT(U)$. We now
describe the function $\profit$.

In the single-dimensional cost setting, each element $e\in U$ is
endowed with a value $\val(e)$ and a size $s(e)$. Values and sizes are
integral and are a priori unknown. The size and value functions extend
to sets of elements as $s(A) = \sum_{e \in A} s(e)$ and $\val(A)
= \sum_{e \in A} \val(e)$. Then the ``profit'' of a subset is given by
$\profit(A) = \val(A) - \cost(s(A))$ where $\cost$ is a non-decreasing convex
function on size: $\cost: \mathbb{Z}^+ \rightarrow
\mathbb{Z}^+$. The following quantities will be useful in our analysis:

 
 
\begin{itemize}
\item The {\em density} of an element, $\rho(e) := \val(e) / s(e)$. We
assume without loss of generality that densities of elements are
unique and denote the unique element with density $\gamma$ by
$e_\gamma$.
\item The {\em marginal cost} function, $\dcost(s) := \cost(s) - \cost(s-1)$. Note that this is an increasing function.
\item The {\em inverse marginal cost} function, $\dcostinv(\rho)$
which is defined to be the maximum size for which an element of
density $\rho$ will have a non-negative profit increment, that is, the
maximum $s$ for which $\rho \geq \dcost(s)$.
\item The {\em density prefix} for a given density $\gamma$ and a set
$T$, $P_\gamma^T := \{e\in T : \rho(e)\ge\gamma\}$, and the partial
density prefix, $\bP_\gamma^T := P_\gamma^T \setminus \{e_\gamma\}$. We
use $P_\gamma$ and $\bP_\gamma$ as shorthand for $P_\gamma^U$ and
$\bP_\gamma^U$ respectively.
\end{itemize}



We will sometimes find it useful to discuss fractional relaxations of
the offline problem of maximizing $\profit$ subject to $\Feas$. To
this end, we extend the definition of subsets of $U$ to allow for
fractional membership. We use $\alpha e$ to denote an
$\alpha$-fraction of element $e$; this has value $\val(\alpha e)
= \alpha \val(e)$ and size $s( \alpha e) = \alpha s(e)$. We say that a
fractional subset $A$ is feasible if its support $\supp(A)$ is
feasible. Note that when the feasibility constraint can be expressed
as a set of linear constraints, this relaxation is more restrictive
than the natural linear relaxation.


Note that since costs are a convex non-decreasing function of size, it may
at times be more profitable to accept a fraction of an element rather
than the whole. That is, $\argmax_{\alpha} \profit(\alpha e)$ may be
strictly less than $1$. For such elements, $\rho(e)
< \dcost(s(e))$. We use $\F$ to denote the set of all such elements:
$\F=\{e\in U : \argmax_{\alpha} \profit(\alpha e) < 1\}$, and $\I =
U\setminus\F$ to denote the remaining elements. Our solutions will
generally approximate the optimal profit from $\F$ by running Dynkin's
algorithm for the classical secretary problem; most of our analysis
will focus on $\I$. Let $\FOPT(T)$ denote the optimal (feasible)
fractional subset of $T\cap\I$ for a given set $T$. Then
$\profit(\FOPT(T)) \geq \profit(\OPT(T\cap\I))$. We use $\FOPT$ as
shorthand for $\FOPT(U)$, and let $s^*$ be the size of this solution.


In the multi-dimensional setting each element has an $\ell$-dimensional
size $s(e) = (s_1(e), \ldots, s_\ell(e))$. The cost function is composed
of $\ell$ different non-decreasing convex functions,
$\cost_i: \mathbb{Z}^+ \rightarrow \mathbb{Z}^+$. The cost of a set of
elements is defined to be $\cost(A)= \sum_{i} \cost_i( s_i(A))$ and
the profit of $A$, as before, is its value minus its cost: $\profit(A)
= v(A) - \cost(A)$.

\subsection{Balanced Sampling}
\label{app:balanced-sampling}
Our algorithms learn the distribution
of element values and sizes by observing the first few
elements. Because of the random order of arrival, these elements form
a random subset of the universe $U$. The following concentration
result is useful in formalizing the representativeness of the sample.
\begin{lemma}
\label{lemma:concentration}
Given constant $c \geq 3$ and a set of elements $I$ with associated
non-negative weights, $w_i$ for $i \in I$, say we construct a random
subset $J$ by including each element of $I$ uniformly at random with
probability $1/2$. If for all $k \in I$,
$w_k \leq \frac{1}{c} \sum_{i \in I} w_i $ then the following
inequality holds with probability at least $0.76$:
\begin{align*}
\sum_{  j \in J} w_j \geq \beta(c) \sum_{i  \in I} w_i, 
\end{align*}
where $\beta(c)$ is a non-decreasing function of $c$ (and furthermore
is independent of $I$).
\end{lemma}

We begin the proof of Lemma~\ref{lemma:concentration} with a restatement of Lemma~1 from~\cite{Feige2005} since it
plays a crucial role in our argument.  Note that we choose a different
parameterization than they do, since in our setting the balance
between approximation ratio and probability of success is different.

\begin{lemma}
\label{lem:balanced-sampling}
Let $X_i$, for $i\geq1$, be indicator random variables for a sequence
of independent, fair coin flips.  Then, for $S_i=\sum_{k=1}^{i}X_k$, we
have
$\Pr [\forall i, \ S_i \geq \floor{i/3}] \geq 0.76.$
\end{lemma}

We now proceed to prove Lemma~\ref{lemma:concentration}.  While we do
not give a closed form for the approximation factor $\lOneApprox$ in
the statement of the lemma, we define it implicitly as
\begin{equation*}
  \lOneApprox
  =\max_{0<\decRate<1}
  \left(\frac{\decRate}{2+\decRate}\right)
  \left(1-\frac{2}{c(1-\decRate)}\right),
\end{equation*}
and give explicit values in Table~\ref{tbl:lemma_one_approxes} for
particular values of $c$ that we invoke the lemma with.

\begin{proof}[Proof of Lemma~\ref{lemma:concentration}]
Our general approach will be to separate our set of weights $\wSet$ into a
``good'' set $\gSet$ and a ``bad'' set $\bSet$.  At a high
level, Lemma~\ref{lem:balanced-sampling} guarantees us that at worst,
we will accept a weight every third time we flip a coin.  So the
good case is when weights do not decrease too quickly; this intuition
guides our definitions of $\gSet$ and $\bSet$.

Let $\decRate\in(0,1)$ be a lower bound on ``acceptable'' rates of
decrease; we tune the exact value later.  Throughout, we use
$\wgt(S)$, $\minWgt(S)$, and $\maxWgt(S)$ to denote the total, minimum, and
maximum weights of a set $S$.  We form $\gSet$ and $\bSet$ as follows.

Initialize $\bSet=\emptyset$ and $i=1$.  Consider $\wSet$ in order of
decreasing weight.  While $\wSet\neq\emptyset$, we repeat the
following.  Let $P$ be the largest prefix such that
$\minWgt(P)\ge\decRate\cdot\maxWgt(P)$.  If $\abs{P}\le2$, move $P$
from $\wSet$ to $\bSet$, i.e. set $\bSet:=\bSet\cup P$ and
$\wSet:=\wSet\setminus P$.  Otherwise, $\abs{P}\ge3$ and we define
$\gSet_i$ to be the $3$ largest elements of $P$; remove them from
$\wSet$ (i.e. set $\wSet:=\wSet\setminus\gSet_i$); and increment $i$
by $1$.  Once we are done, we define $\gSet=\cup_{i}\gSet_i$.

First, we show that the total weight in $\bSet$ cannot be too large.
Note that we add at most $2$ elements at a time to $\bSet$; and when
we do add elements, we know that all remaining elements in $\wSet$ (and
hence all later additions to $\bSet$) are smaller by more than a
factor of $\decRate$.  Thus, we can see that
\begin{equation*}
  \wgt(\bSet)
  \le\sum_{i\ge0}2\decRate^{i}\cdot\maxWgt(\bSet)
  \le\frac{2\maxWgt(\bSet)}{1-\decRate}
  \le\frac{2\wgt(\wSet)}{c(1-\decRate)},
\end{equation*}
by our assumption that no individual weight is more than $\wgt(\wSet)/c$.

Next, we show that with probability at least $0.76$, we can lower bound
the fraction of weight we keep from $\gSet$.  Consider applying
Lemma~\ref{lem:balanced-sampling}, flipping coins first for the
weights in $\gSet$ in order by decreasing weight.  Note that by the
time we finish flipping coins for $\gSet_i$, we must have added at
least $i$ weights to $J$; hence the $i^{\text{th}}$ weight we add to
$J$ must have value at least $\minWgt(\gSet_i)$.  On the other hand,
we know that
\begin{equation*}
  \wgt(\gSet_i) 
  \le 2\maxWgt(\gSet_i) + \minWgt(\gSet_i) 
  \le \left(\frac{2}{\decRate}+1\right)\minWgt(\gSet_i),  
\end{equation*}
and so summing over $i$ we can see that elements we accept have total
weight $\wgt(J)\ge(\frac{\decRate}{2+\decRate})\wgt(\gSet)$.

Combining our bounds for the weights of $\gSet$ and $\bSet$, we can
see that with probability $0.76$ the elements we accept have weight
\begin{equation*}
  \wgt(J)
  \ge\left(\frac{\decRate}{2+\decRate}\right)\wgt(\gSet)
  =\left(\frac{\decRate}{2+\decRate}\right)(\wgt(\wSet)-\wgt(\bSet))
  \ge\left(\frac{\decRate}{2+\decRate}\right)\left(1-\frac{2}{c(1-\decRate)}\right)\wgt(\wSet);
\end{equation*}
optimizing the above with respect to $\decRate$ for a fixed $c$ gives
the claimed result.  Note that for each fixed
$\decRate\in(0,1)$ our approximation factor is increasing in $c$, and so
 the  optimal value $\lOneApprox$ must be increasing in $c$ as well.
\end{proof}
\begin{figure}
  \begin{center}
    \begin{tabular}{|c|c|c|}
      \hline
      $c$ given & $\decRate$ chosen & $\lOneApprox$ achieved \\
      \hline
      $111$ & $0.84$ & $\approx0.262$\\
      $15/2$ & $0.46$ & $\approx0.094$\\
      $\lOneSize$ & $\lOneTune$ & $\approx\lOneRatioTable$\\
      \hline
    \end{tabular}
  \end{center}
  \caption{Some specific values of approximation ratio achieved by
    Lemma~\ref{lemma:concentration}.}
  \label{tbl:lemma_one_approxes}
\end{figure}

\section{Unconstrained Profit Maximization}
\label{sec:unconstrained}
\label{sect:unconstrained}
We begin by developing an algorithm for the unconstrained
version of the generalized secretary problem with $\Feas=2^U$, which
already exhibits some of the challenges of the general setting. Note
that this setting captures as a special case the knapsack secretary
problem of \cite{babaioff2007knapsack} where the goal is to maximize
the total value of a subset of size at most a given bound. In fact in
the offline setting, the generalized secretary problem is very similar
to the knapsack problem. If all elements have the same (unit) size,
then the optimal offline algorithm orders elements in decreasing order
of value and picks the largest prefix in which each element
contributes a positive marginal profit. When element sizes are
different, a similar approach works: we order elements by density this
time, and note that either a prefix of this ordering or a single
element is a good approximation (much like the greedy
$2$-approximation for knapsack). 

\begin{algorithm}
\caption{Offline algorithm for single-dimensional $(\profit,2^U)$}
\label{Algorithm:greedy}
\begin{algorithmic}[1]
\STATE  Initialize set $\I \leftarrow \{ a \in U \ \mid  \   \rho(a) \geq \dcost \left( s(a) \right)  \}$
\STATE Initialize selected set $\alg \leftarrow \emptyset$
\STATE Sort $\I$ in decreasing order of density. 
\FOR {$i \in \I$}   
     \IF {marginal profit for the $i$th element, $\profit_\alg (i)  \geq 0$ } 
     
           \STATE $\alg \leftarrow \alg \cup \{ i \}$
    \ELSE 
         \STATE Exit loop.
    \ENDIF
\ENDFOR
\STATE Set $m := \argmax_{\{a \in U\}} \profit(a)$ 
\COMMENT{the most profitable element}
\IF {$\profit(m) > \profit(\alg)$}
        \STATE Output set $\{ m \}$
 \ELSE 
         \STATE Output set $\alg$       
\ENDIF

\end{algorithmic}
\end{algorithm}

Precisely, we show that $|\OPT\cap\F|\le 1$, and we can therefore
focus on approximating $\profit$ over the set $\I$. Furthermore, let
$\alg(U)$ denote the greedy subset obtained by
Algorithm~\ref{Algorithm:greedy}, which considers elements in $\I$ in
decreasing order of density and picks the largest prefix where every
element has nonnegative marginal profit. The following lemma implies
that one of $\alg(U)$ or the single best element is a $3$-approximation to
$\OPT$.


\begin{lemma}
\label{lemma:offline-knapsack}
We have that
$\profit(\OPT) \leq \profit(\FOPT) + \max_{e \in U} \profit(e) \leq \profit(\alg(U)) + 2 \max_{e \in U} \profit(e)$. Therefore the greedy offline algorithm 
(Algorithm~\ref{Algorithm:greedy})
achieves a 3-approximation for $(\profit,2^U)$.
\end{lemma}

\begin{proof}

We first show that $\OPT$ has at most one element from $\Fr$. For 
contradiction, assume that $\OPT$ contains at least two elements $f_1,
f_2\in\Fr$. Since densities are unique, without loss of generality, we assume $\rho(f_1) > 
\rho(f_2)$. 

Recall that $\Fr$ is precisely the set of elements for which it is
optimal to accept a strictly fractional amount.  Let $\alpha_1,
\alpha_2 < 1$ be the optimal fractions for $f_1$ and $f_2$,
i.e. $\argmax_{\alpha} \profit(\alpha f_1 ) = \alpha_1$ and 
$ \argmax_{\alpha} \profit(\alpha f_2 )=  \alpha_2$.  Then
adding a fractional amount of any element with density at most $\rho(f_1)$
to $\{\alpha_1 f_1\}$ results in strictly decreased profit. But this
implies $\profit(\OPT)< \profit(\OPT \setminus\{f_2\})$, contradicting
the optimality of $\OPT$.

 
Let $f$ be the unique element in $\OPT \cap \Fr$. By subadditivity, we get
that $\profit(\OPT\setminus \{f\}) + \profit(f) \geq \profit(\OPT) $. 
Since $\OPT\setminus \{ f \}\subseteq \I$, we have  
$\profit(\FOPT) \geq \profit(\OPT\setminus \{ f \}) $. 
In the rest of the proof we focus on approximating $\profit(\FOPT)$.


Note that $\FOPT$ is a fractional density prefix of $\I$. So let 
$\FOPT = P_\rhoc \cup \{\alpha e\}$, for some $e$ and $\alpha < 1$. 
The subadditivity of $\profit$ implies $\profit(\FOPT) \leq
\profit(\preflarge) + \profit(\{\alpha e\})$.
Note that Algorithm~\ref{Algorithm:greedy} selects 
$\alg=\preflarge$, and that $\profit(\{\alpha e\}) \leq
\profit(e)$ since $e \in\I$.  


Hence, combining the above inequalities we get $\profit(\alg) +
\profit(e) + \profit(f) \geq \profit(\OPT) $. This in turn proves the required claim.
\end{proof}



The offline greedy algorithm suggests an online solution as well. In
the case where a single element gives a good approximation, we can use
the classical secretary algorithm to get a good competitive ratio. In
the other case, in order to get good competitive ratio, we merely need
to estimate the smallest density, say $\rhoc$, in the prefix of
elements that the offline greedy algorithm picks, and then accept
every element that exceeds this threshold.





We pick an estimate for $\rhoc$ by observing the first few elements of
the stream $U$. Note that it is important for our estimate of $\rhoc$
to be no smaller than $\rhoc$. In particular, if there are many
elements with density just below $\rhoc$, and our algorithm uses a
density threshold less than $\rhoc$, then the algorithm may be fooled
into mostly picking elements with density below $\rhoc$ (since
elements arrive in random order), while the optimal solution picks
elements with densities far exceeding $\rhoc$. We now describe how to
pick an overestimate of $\rhoc$ which is not too conservative, that is,
such that there is still sufficient profit in elements whose densities exceed
the estimate.

\label{sect:online_unconstrained}
In the remainder of this section, we assume that every element has
profit at most $\frac{1}{\const[1]+1} \profit(\OPT)$ for an appropriate
constant $\const[1]$, to be defined later. 
(If this does not hold, the classical secretary algorithm obtains an expected profit of at least $\frac{1}{ e ( \const[1]+1 )} \profit(\OPT)$). Then Lemma~\ref{lemma:offline-knapsack} implies 
$\profit(\FOPT) \geq \left( 1 - \frac{1}{(\const[1]+1)} \right) \profit(\OPT)$, $\max_{e\in U}\profit(e)\le \frac {1}{\const[1]} \profit(\FOPT)$, and
$\profit(\alg(U))\geq \left(1 - \frac{1}{\const[1]} \right) \profit(\FOPT)$. 

We divide the stream $U$ into two parts $X$ and $Y$, where $X$ is a
random subset of $U$. Our algorithm unconditionally rejects elements
in $X$ and extracts a density threshold $\guess$ from this set. Over
the remaining stream $Y$, it accepts an element if and only if its
density is at least $\guess$ and if it brings in strictly positive
marginal profit. Under the assumption of small element profits we can
apply Lemma \ref{lemma:concentration} to show that
$\profit(X \cap \alg(U))$ is concentrated and is a large enough
fraction of $\profit(\OPT)$. This implies that with high probability
$\profit(X \cap \alg(U))$ (which is a prefix of $\alg(X)$) is a
significant fraction of $\profit(\alg(X))$. Therefore we attempt to
identify $X \cap \alg(U)$ by looking at profits of prefixes of $X$.

\begin{algorithm}
\caption{Online algorithm for single-dimensional $(\profit,2^U)$}
\label{Algorithm:weighted}
\begin{algorithmic}[1]
\STATE  With probability $1/2$ run the classic secretary algorithm
to pick the single most profitable element else execute the following steps.
\STATE Draw $k$ from $\textrm{Binomial}(n,1/2)$.
\STATE Select the first $k$ elements to be in the sample $X$. Unconditionally reject these elements.
\STATE Let $\guess$ be largest density such that $\profit(P_\guess^X) \geq \eone \left(1 - \frac{1}{\const[1]} \right) \profit(\FOPT(X)) $ for constants $\eone$ and $\const[1]$ to be specified later.
\STATE Initialize selected set $O \leftarrow \emptyset$.
\FOR {$i \in Y = U \setminus X$}
    \IF { $\profit(O \cup \{i\}) - \profit(O) \geq 0$ and $\rho(i) \geq \guess$ and $i \notin \F$}
           \STATE $O \leftarrow O \cup \{ i \}$
    \ELSE 
         \STATE Exit loop.
    \ENDIF
\ENDFOR
\end{algorithmic}
\end{algorithm}




We will need the following lemma about $\alg()$.
\begin{lemma}
\label{lemma:monotone-integral}
For any set $S$, consider subsets $A_1, A_2 \subseteq \alg(S)$. If $A_1 \supseteq A_2$, then $\profit(A_1) \geq \profit(A_2)$. In other words, $\profit$ is monotone-increasing when restricted to $\alg(S)$ for all $S \subset U$.
\end{lemma}

\begin{proof}
We observe that the fractional greedy algorithm sorts its input $S$ by
decreasing order of density, and $\alg(S)$ consists of the top $|\alg(S)|$
elements under that ordering.  Since $\FOPT(S)$ contains each element
in $\alg(S)$ in its entirety, we can see that $\FOPT(B) = B$ for any
subset $B$ of $\alg(S)$.  So for $A_2\subseteq A_1 \subseteq \alg(S)$,
we have that $\FOPT(A_2)=A_2\subseteq A_1=\FOPT(A_1)$; by the
optimality of $\FOPT$, this implies that $\profit(A_2)\le\profit(A_1)$
as claimed.
\end{proof}



We define two good events. $E_1$ asserts that $X\cap\alg(U)$ has high
enough profit. Our final output is the set $P_{\guess}^Y$. $E_2$
asserts that the profit of $P_{\guess}^Y$ is a large enough fraction
of the profit of $P_{\guess}$. Recall that $\alg(U)$ is a density prefix, say  $P_{\rhoc}$, and
so $X \cap \alg(U) = P_{\rhoc}^X$. We define the event $E_1$ as follows.
\begin{align*}
 E_1: &  \ \ \profit(P_\rhoc^X)   > \eone \ \profit(P_\rhoc)
 \end{align*}
where $\eone$ is a constant to be specified later. Conditioned on
 $E_1$, we have $\profit(P_\rhoc^X)
 > \eone \left( 1-1/\const[1] \right)\profit(\FOPT)\ge \eone 
\left( 1-1/\const[1] \right)\profit(\FOPT(X))$. Note
 that threshold $\guess$, as selected by
 Algorithm \ref{Algorithm:weighted}, is the largest density such that
 $\profit(P_\guess^X) \geq \eone \left(1 -
 1/\const[1] \right) \profit(\FOPT(X))$. Therefore, $E_1$ implies
 $\guess \geq \rhoc$, and we have the following lemma.


\begin{lemma}
Conditioned on $E_1$, $O=P_\guess \cap Y \subseteq \alg(U)$.
\end{lemma}

\noindent
On the other hand, $ P_\guess^X \subseteq P_\guess \subset \alg(U)$ along with Lemma~\ref{lemma:monotone-integral} implies
\begin{align*}
\profit(P_\guess) & \geq \profit(P_\guess^X)
\geq \eone \left(1-1/\const[1] \right)\profit(\FOPT(X)) 
\geq \eone\left(1-1/\const[1] \right)\profit(P_\rhoc^X)
\geq \eone^2 \left(1-1/\const[1] \right)^2\profit(\FOPT)
\end{align*}
where the second inequality is by the definition of $\guess$, the third by optimality and the last is obtained by applying $E_1$ and $\alg(U) \geq \left( 1 - 1/\const[1] \right) \FOPT$.

We define $\rhob$ to be the largest density such that $\profit(P_\rhob) \geq \eone^2 \left(1-\frac{1}{\const[1]}\right)^2 \profit(\FOPT)$. Then $\rhob \geq \guess $, which implies $P_\rhob \subseteq P_\guess $ and the following lemma.
\begin{lemma}
Event $E_1$ implies $O \supseteq Y \cap P_\rhob$. 
\end{lemma} 
\noindent Based on the above lemma, we define event $E_2$ for an appropriate
constant $\etwo$ as follows
\begin{align*}
E_2 : & \,\, \profit(P_\rhob^Y) \geq \etwo \profit(P_\rhob).
\end{align*}
Conditioned on events $E_1$ and $E_2$, and using Lemma~\ref{lemma:monotone-integral} again, we get 
\begin{align*}
\profit(O) \ge \profit(P_\rhob^Y)  \ge \etwo \beta^2(1-1/\const[1])^2 \profit(\FOPT).
\end{align*}
To wrap up the analysis, we show that $E_1$ and $E_2$ are high probability events.
\begin{lemma}
\label{lemma:chernoff-unconstrained}
If no element of $U$ has profit more than $\frac{1}{113} \profit(
\OPT) $, then $\textrm{Pr}[E_1 \ \wedge \ E_2] \geq 0.52$, where
$\beta=0.262$ and $\beta'=0.094$.
\end{lemma}

\begin{proof}
We show that $\Pr[ E_1 ],\Pr[ E_2] \ge
0.76$; the desired  inequality $\Pr[E_1 \ \wedge \ E_2 ] \geq
0.52$ then follows by the union bound.

In the following, we assume the elements are sorted in decreasing
order of density. Denote the profit increment of the $i^{\text{th}}$
element by $\incprofit(i) = \profit(\prefi[i]) -
\profit(\prefi[i-1])$; this extends naturally to sets $A \subseteq
\alg(U)$ by setting $\incprofit(A)=\sum_{i\in A} \incprofit(i)$. By
the subadditivity of $\profit$, we have $\profit(A) \geq
\incprofit(A)$ for all $A\subseteq \alg(U)$, with equality at
$\alg(U)$.

We apply Lemma \ref{lemma:concentration} with $P_\deltac$ as the fixed
set $I$ and $P_\deltac^X = U \cap P_\deltac$ as the the random set
$J$. The weights in Lemma \ref{lemma:concentration} correspond to
profit increments of the elements. Note that $P_\deltac=\alg(U)$; so
we know both that elements in $P_\deltac$ have non-negative profit
increments, and $\profit(P_\deltac)\ge\profit( \OPT)-2\max_{e\in U}\profit(e)$.
Hence if no element has profit exceeding a $1/113$-fraction of
$\profit( \OPT)$, we get that any element $e_i\in P_\deltac$ has
profit increment
$\incprofit(i)\le\profit(e_i)\le1/(113-2)\profit(P_\deltac)=1/111\incprofit(P_\deltac)$.
Hence we can apply Lemma~\ref{lemma:concentration} to get
$\incprofit(P_\deltac^X) \geq 0.262\incprofit(P_\deltac)$ with
probability at least $0.76$, and so
the event $E_1$ holds with probability at least $0.76$.

By our definition of $\deltab$, the profit of $P_\deltab$ is at least
$\eone^2(1-1/\const[1])^2(1-1/(\const[1]+1)) \profit(\OPT)$; substituting in the
specified values of $\beta$ and $\const[1]$ give us that no element in
$P_\deltab$ has profit increment more than $2/15
\ \incprofit(P_\deltab) $. Thus, applying
Lemma~\ref{lemma:concentration} implies $\Pr[ E_2 ] \ge 0.76$ with
$\eone^\prime = 0.094$.
\end{proof}




Putting everything together we get the following theorem.

 \begin{theorem}
   \label{theorem:weighted}
   Algorithm \ref{Algorithm:weighted} achieves a competitive ratio of $616$
   for $(\profit, 2^U)$ using $\const[1]=112$ and $\beta=0.262$. 
 \end{theorem}

\begin{proof}
If there exists an element with profit at least $\frac{1}{113}
\profit(\OPT(U))$, the classical secretary algorithm (Step~$1$) gives a
competitive ratio of $\frac{1}{113e}\ge\frac{1}{308}$.  Otherwise,
using Lemma~\ref{lemma:chernoff-unconstrained}, with $\beta'=0.094$, we have
$\mathbb{E}[\profit(O)]\geq \mathbb{E}[\profit(O) \mid E_1 \wedge E_2]
\Pr [ E_1\wedge E_2 ] \geq 0.52\etwo \eone^2(1-1/\const[1])^2
\profit(\FOPT) \ge 0.52\etwo \eone^2(1-1/\const[1])^2 \left(1 - 1/(\const[1]+1 )\right)
\profit(\OPT)\ge\frac{1}{307}\profit(\OPT)$.  Since we flip a fair
coin to decide whether to output the result of running the classical
secretary algorithm, or output the set $O$, we achieve a
$2\max\{308,307\}=616$-approximation to $\profit(\OPT)$ in
expectation (over the coin flip).
\end{proof}

\section{Matroid-Constrained Profit Maximization}
\label{sec:constrained}
We now extend the algorithm of Section~\ref{sec:unconstrained} to the
setting $(\profit, \Feas)$ where $\Feas$ is a matroid constraint. In
particular, $\Feas$ is the set of all independent sets of a matroid
over $U$. We skip a precise definition of matroids and will only use the following facts: $\Feas$ is a downward closed feasibility constraint and there exist an exact offline and an $O(\sqrt{\log r})$ online algorithm for
$(\Phi,\Feas)$, where $\Phi$ is a sum-of-values objective and $r$ is the rank of the matroid. 

In the unconstrained setting, we showed that there always exists either a density prefix or a single element with near-optimal profit. So in the online setting it sufficed to determine the density threshold for a good prefix. In constrained settings this is no longer true, and we need to develop new techniques. Our approach is to develop a general reduction from the $\profit$ objective to two different sum-of-values type objectives over the same feasibility constraint. This allows us to employ previous work on the $(\Phi,\Feas)$ setting; we lose only a constant factor in the competitive ratio. We will first describe the reduction in the offline setting and then extend it to the online algorithm using techniques from Section~\ref{sec:unconstrained}.

\paragraph{Decomposition of $\profit$.}

For a given density $\gamma$, we define the {\em shifted density
function} $\sdf_\gamma( )$ over sets as $ \sdf_\gamma(A)
:= \sum_{e \in A} \left(\rho(e) - \gamma \right) s(e) $ and the {\em
fixed density function} $\fdf_\gamma( )$ over sizes as
$ \fdf_\gamma(s) := \gamma s - \cost(s)$. For a set $A$ we use
$\fdf_\gamma(A)$ to denote $\fdf_\gamma(s(A))$. It is immediate that
for any density $\gamma$ we can split the profit function as
$\profit(A) = \sdf_\gamma(A) + \fdf_\gamma(A)$. In particular
$\profit(\OPT)= \sdf_{\gamma}(\OPT)+\fdf_{\gamma}(\OPT)$. Our goal
will be to optimize the two parts separately and then return the
better of them.

Note that the function $\sdf_{\gamma}$ is a sum of values function
where the value of an element is defined to be $(\rho(e) - \gamma)
s(e)$. Its maximizer is a subset of $P_\gamma$, the set of elements
with nonnegative shifted density $\rho(e)-\gamma$. In order to ensure
that the maximizer, say $A$, of $\sdf_{\gamma}$ also obtains good
profit, we must ensure that $\fdf_{\gamma}(A)$ is nonnegative, and
therefore $\profit(A)\ge \sdf_{\gamma}(A)$. This is guaranteed for a
set $A$ as long as $s(A)\le \dcostinv(\gamma)$.

Likewise, the function $\fdf_{\gamma}$ increases as a function of size
$s$ as long as $s$ is at most $\dcostinv(\gamma)$, and decreases
thereafter. Therefore, in order to maximize $\fdf_\gamma$, we merely
need to find the largest (in terms of size) feasible subset of size no
more than $\dcostinv(\gamma)$. As before, if we can ensure that for
such a subset $\sdf_\gamma$ is nonnegative (e.g. if the set is a
subset of $P_\gamma$), then the profit of the set is no smaller than
its $\fdf_{\gamma}$ value. This motivates the following definition of
``bounded'' subsets:
\begin{definition}
\label{def:bounded}
Given a density $\gamma$ a subset $A \subseteq U$ is said to be $\gamma$-bounded if $A\subseteq P_\gamma$ and $ s(A) \leq \dcostinv( \gamma)$.
\end{definition}
 
We begin by formally proving that the function $\fdf_\gamma$ increases
as a function of size $s$ as long as $s$ is at most
$\dcostinv(\gamma)$.
\begin{lemma}
  \label{lem:monotone}
  If density $\gamma$ and sizes $s$ and $t$ satisfy $s \leq t \leq \dcostinv(\gamma)$, then $ \fdf_\gamma(s)
  \leq \fdf_\gamma(t) $.
\end{lemma}

\begin{proof}
Since $\cost(\cdot)$ is convex, we have that its marginal, $\dcost()$,
is monotonically non-decreasing. Thus we get the following chain of
inequalities,
  \begin{align*}
    \cost(t) - \cost(s) &= \sum_{z=s+1}^t \dcost(z) \leq (t - s)
    \times \dcost(t) \leq (t - s) \gamma.
  \end{align*}
The last inequality follows from the assumption that $t$ is no more than $\dcostinv (\gamma) $ and hence $ \dcost(t) \leq \gamma$. By
 definition of $\fdf_\gamma()$ we get the desired claim.
\end{proof}

\begin{proposition}
\label{lem:profit-lb}
For any $\gamma$-bounded set $A$, $\profit(A)  \geq \sdf_\gamma(A)$ and  $\profit(A)  \geq \fdf_\gamma(A)$.
\end{proposition}

\begin{proof}
 Since $\profit(A) = \sdf_\gamma(A) + \fdf_\gamma(A)$, it is
 sufficient to prove that $ \sdf_\gamma(A)$ and $\fdf_\gamma(A)$ are
 both non-negative. The former is clearly non-negative since all
 elements of $A$ have density at least $\gamma$. Lemma
 \ref{lem:monotone} implies that the latter is non-negative by taking
 $t = s(A)$ and $s = 0$.
\end{proof}

\noindent
For a density $\gamma$ and set $T$ we define $H_\gamma^T$ and $G_\gamma^T$ as follows. (We write $H_\gamma$ for $H_\gamma^U$ and $G_\gamma$ for $G_\gamma^U$.)
 \begin{align*}
    H_\gamma^T &= \argmax_{H \in \Feas, H \subset P_\gamma^T} \sdf_{\gamma}(H)
&
    G_\gamma^T &= \argmax_{G \in \Feas, G \subset
      \bP_\gamma^T } s(G) 
  \end{align*}

Following our observations above, both $H_\gamma$ and $G_\gamma$ can
be determined efficiently (in the offline setting) using standard
matroid maximization. However, we must ensure that the two sets are
$\gamma$-bounded. Further, in order to compare the performance of
$G_{\gamma}$ against $\OPT$, we must ensure that its size is at least
a constant fraction of the size of $\OPT$. We now show that there
exists a density $\gamma$ for which $H_\gamma$ and $G_\gamma$ satisfy
these properties.

Once again, we focus on the case where no single element has high enough profit by itself. Recall that $\FOPT$ denotes $\FOPT(\I)$, $s^*$ denotes the size of this set and $\bP_\gamma$ denotes $P_\gamma \setminus \{ e_\gamma \}$. Before we proceed we need the following fact about the fractional
optimal subset $\FOPT$. 

\begin{lemma}
\label{lem:opt_frac}
 If $\FOPT$ has an element of density $\gamma$ then $s^*$ is at most
 $\dcostinv(\gamma)$.
\end{lemma}
\begin{proof}
The proof is by contradiction. Say $s^*$ is more than
$\dcostinv(\gamma)$. Recall that $e_\gamma$ denotes the element 
with density $\gamma$. We show that in such conditions reducing the
fractional contribution of $e_\gamma$, say by $\epsilon$, increases profit
giving us a better fractional solution. This is turn contradicts the
optimality of $\FOPT$.

Write $s = s(e_\gamma)$ and note that
\begin{align*} 
\profit(\FOPT) & = v(\FOPT) - \cost(s^*) \\ & = \left[ \left( v(\FOPT)
  - \gamma \times \epsilon s \right) - \cost(s^* - \epsilon s )
  \right] - \left[ \cost(s^*) - \cost(s^* - \epsilon s ) - \gamma
  \times \epsilon s \right].
\end{align*}
If $\epsilon$ is such that $s^* -\epsilon > \dcostinv(\gamma) $, then we have $ \cost(s^*) - \cost(s^* -
\epsilon s ) > \gamma \times \epsilon s $; thus we get that the term 
$ \left[ \cost(s^*) - \cost(s^* - \epsilon s ) - \gamma \times \epsilon s
  \right]$ is positive which proves the claim.
\end{proof}

\begin{definition}
Let $\rhoc$ be the largest density such that $\preflarge$ has a feasible set of size at least $s^*$.
\end{definition}

\noindent
We now state a useful property of $\rhoc$. 
\begin{lemma}
  \label{lem:rhoc-indep}
Any feasible set in $\bP_\rhoc$ is $\rhoc$-bounded and has size less than $s^*$. Moreover for any density $\gamma > \rhoc$ all feasible subsets of $P_\gamma$ are $\gamma$-bounded. 
\end{lemma}

\begin{proof}
By definition, the size of any feasible set contained in $P_\rhoc
\setminus \{ e_\rhoc \}$ is no more than $s^*$.

We will show that $s^* \leq \dcostinv( \rhoc)$. Then the first part of the lemma
follows immediately. For the second part we have $\gamma > \rhoc$ and hence 
$\dcostinv(\rhoc) \leq \dcostinv(\gamma) $. Overall a size bound of $s^*$ also implies that 
feasible sets in $P_\gamma$ satisfy the size requirement for being
$\gamma$-bounded. Hence we get that all feasible sets in $P_\rhoc
\setminus \{ e_\rhoc \}$ are $\rhoc$-bounded and all feasible sets
in $P_\gamma$ are $\gamma$-bounded.

The size of $\FOPT$ is at most the size of its support. Thus the
support of $\FOPT$ is a feasible set of size at least $s^*$. By
definition, $\rhoc$ is the largest density such that $P_\rhoc$
contains a feasible set of size $s^*$. Hence we get that $\FOPT$
contains an element of density less than or equal to $\rhoc$. Applying
Lemma \ref{lem:opt_frac} we get $s^* \leq \dcostinv(\rhoc)$ and the
lemma follows.
\end{proof}

\noindent
The following is our main claim of this section. 
\begin{lemma}
  \label{lem:goodprefix} 
For any density $\gamma>\rhoc$,
$\profit(\OPT(\bP_\gamma )) \le \profit(H_\gamma)+\profit(G_\gamma)$. Furthermore, $\profit(\OPT) \leq \profit(H_\rhoc) + \profit(G_\rhoc) + 2\max_{e\in U} \profit(e)$.
\end{lemma}

\begin{proof}
Let $P$, $H$, and $G$ denote $\preflarge$, $H_{\rhoc}$, and
$G_{\rhoc}$ respectively. As in the unconstrained setting, there can
be at most one element in the intersection of $\OPT$ and $\Fr$ 
(see proof of Lemma \ref{lemma:offline-knapsack}). Note that $\profit()$ is subadditive, hence
$\profit( \OPT \cap \I) + \max_{e \in U } \profit(e) \geq
\profit(\OPT)$.  In the analysis below we do not consider elements
present in $\Fr$ and show that $\profit(H) + \profit(G) + \max_{e \in
  U } \profit(e) \geq \profit(\OPT \cap \I)$.  This in turn
establishes the second part of the Lemma.


For ease of notation we denote $e_\rhoc$ as $e'$. Without loss of
generality, we can assume that $H$ does not contain $e'$ since
$\rho(e') - \rhoc = 0$.  Therefore set $H$ is contained in $P
\setminus \{ e' \}$. By Lemma \ref{lem:rhoc-indep} we get that $H$ is
$\rhoc$-bounded.

Note that by definition $G$ does not contain $e'$, hence its size is
at most $s^*$.  Also, $G \cup \{e' \}$ is no smaller than the
maximum-size feasible subset contained in $P$. So, by definition of
$\rhoc$, we also have $s(G) + s(e') \geq s^*$. Thus there exists
$\alpha < 1$ such that the fractional set $F = G \cup \{ \alpha e' \}
$ has size exactly equal to $s^*$.

 

Next we split the profit of $\FOPT$ into two parts, and bound the first
by $\sdf_\rhoc(H)$ and the second by $\fdf_\rhoc (F)$:
\begin{align*}
  \profit(\FOPT) &= \sdf_\rhoc( \FOPT) + \fdf_\rhoc(s^*) 
  \leq \sdf_\rhoc(H) + \fdf_\rhoc(F) .
\end{align*}
\noindent Note that we can drop elements which have a negative contribution to
the sum to get
\begin{align*}
\sdf_\rhoc(\FOPT) & \leq \sdf_\rhoc(\FOPT \cap \preflarge ) \\
& \leq \sdf_\rhoc \left( \supp(\FOPT) \cap \preflarge \right) \\
& \leq \sdf_\rhoc(H) \\
& \leq \profit(H).
\end{align*}
The second inequality follows since we can only increase the value of
a subset by ``rounding up'' fractional elements. The third inequality
follows from the optimality of $H$, and the fourth from the fact that
it is $\rhoc$-bounded.

To bound the second part we note that $s(F) = s^*$, hence
$\fdf_\rhoc(F) = \rhoc s^* - \cost(s^*)$. Elements in $F$ have density
no less than $\rhoc$ and its size is bounded above by
$\dcostinv(\rhoc)$, hence it is a $\rhoc$-bounded set implying that
$\profit(F) \geq \fdf_\rhoc(F) $.  Note that $F = G \cup \{ \alpha e'
\}$, and by sub-additivity of $\profit()$ we have $\profit(G) +
\profit(\alpha e' ) \geq \profit(F) $. Moreover $e' \in \I$ implies
$\profit(\alpha e' ) \leq \profit(e')$ and hence we get
\begin{align*}
\profit(\FOPT) & \leq \profit(H) + \fdf_\gamma(F) \leq \profit(H) + \profit(G) + \profit(e'),
\end{align*}
which proves the second half of the lemma.

The first half of the lemma follows along similar lines. We 
have the standard decomposition, $\profit(\OPT( \bar{P}_\gamma) ) =
\sdf_\gamma(\OPT(\bar{P}_\gamma)) + \fdf_\gamma(\OPT(
\bar{P}_\gamma))$. By definition, $H_\gamma$ is the
constrained maximizer of $\sdf_\gamma$, hence we get  
$\sdf_\gamma(\OPT( \bar{P}_\gamma) ) \leq \sdf_\gamma(H_\gamma)$. We 
note that all feasible sets in $P_\gamma$ are $\gamma$-bounded, 
for density $\gamma > \rhoc$ (Lemma \ref{lem:rhoc-indep}). Hence, by Lemma \ref{lem:monotone}, $\fdf_\gamma$ strictly increases with size when restricted to feasible sets in $\bar{P}_\gamma$. $G_\gamma$ is the largest such set, hence we get $\fdf_\gamma(G_\gamma) \geq \fdf_\gamma(\OPT(\bar{P}_\gamma))$. $H_\gamma$ and $G_\gamma$ are $\gamma$-bounded and hence by Proposition \ref{lem:profit-lb} we have $\profit(H_\gamma) \geq \sdf_\gamma(H_\gamma)$ and
$\profit(G_\gamma) \geq \fdf_\gamma(G_\gamma)$. This establishes the lemma.
\end{proof}

This lemma immediately gives us an offline approximation algorithm for
$(\profit,\Feas)$: for every element density $\gamma$, we find the
sets $H_{\gamma}$ and $G_{\gamma}$; we then output the best (in terms of profit) of these sets or the best individual element. We obtain the following theorem:

\begin{theorem}
\label{theorem:matroid_offline}
Algorithm~\ref{Algorithm:offline-matroid} $4$-approximates
$(\profit,\Feas)$ in the offline setting.  
\end{theorem}

\begin{algorithm}
  \caption{Offline algorithm for single-dimensional $(\profit, \Feas)$}
\label{Algorithm:offline-matroid}
\begin{algorithmic}[1]
\STATE Set $A_{\max{}} \leftarrow \argmax_{H \in \{ H_\gamma\}_\gamma } \profit(H)$
\STATE Set $B_{\max{}} \leftarrow \argmax_{ G \in \{G_\gamma \}_\gamma } \profit(G)$ 
\STATE Set $ e_{\max{}} \leftarrow \argmax_{e \in
  U} \profit(e) $ 
\STATE Assign $\alg(U) \leftarrow \argmax_{ S \in \{ A_{\max{}}, B_{\max{}}, e_{\max{}} \}} \profit(S)$
\end{algorithmic}
\end{algorithm}

\paragraph{The online setting.}
Our online algorithm, as in the unconstrained case, uses a sample $X$
from $U$ to obtain an estimate $\tau$ for the density $\rhoc$. Then
with equal probability it applies the online algorithm for
$(\sdf_\tau,\Feas)$ on the remaining set $Y\cap P_\tau$ or the online
algorithm for $(s,\Feas)$ (in order to maximize $\fdf_\tau$) on $Y\cap
P_\tau$. Lemma~\ref{lem:goodprefix} indicates
that it should suffice for $\tau$ to be larger than $\rhoc$ while
ensuring that $\profit(\OPT(P_\tau))$ is large enough. As in
Section~\ref{sec:unconstrained} we define the density $\rhob$ as the
upper limit on $\tau$, and claim that $\tau$ satisfies the required
properties with high probability.

\begin{algorithm}
  \caption{Online algorithm for single-dimensional $(\profit,\Feas)$ }
  \label{Algorithm:online-matroid}
\begin{algorithmic}[1]
\STATE Draw $k$ from $\textrm{Binomial}(n,1/2)$.  
\STATE Select the first $k$ elements to
be in the sample $X$. Unconditionally reject these elements.  \STATE
Toss a fair coin.  \IF {Heads} \STATE Set output $O$ as the first
element, over the remaining stream, with profit higher than $\max_{e \in X} \profit(e)$.
\ELSE \STATE Determine $\alg(X)$ using the offline Algorithm
\ref{Algorithm:offline-matroid}.  \STATE Let $\eone$ be a specified
constant and let $\guess$ be the highest density such that
$\profit(\alg(\prefguess^X)) \geq \frac{\eone}{16}\profit(\alg(X))$.
\STATE \label{step:online-start} Toss a fair coin.  \IF {Heads}
\STATE \label{step:redn-1} Let $O_1$ be the result of executing an
online algorithm for $(\sdf_\guess,\Feas)$ on the subset $P^Y_\tau$ of the remaining
stream with the objective function
\[\sdf_{\guess}(A) = \sum_{e \in A} (\rho(e) - \guess) s(e)\]
\STATE Set $O \leftarrow \emptyset$.  \FOR {$e \in
  O_1$} \label{step:for-1} \IF {$\profit(O \cup \{e\}) - \profit(O)
  \geq 0$} \STATE Set $O \leftarrow O \cup \{e\}$.  \ENDIF \ENDFOR
\STATE Output $O$.  \ELSE \STATE\label{step:redn-2} Let $O_2$ be the result
of executing an online algorithm for $\Feas$ on the subset $P^Y_\tau$
of the remaining stream with objective function $s()$.  \STATE
Set $O \leftarrow \emptyset$.  \FOR {$e \in O_2$} \label{step:for-2}
\IF {$\profit(O \cup \{e\}) - \profit(O) \geq 0$} \STATE Set $O
\leftarrow O \cup \{e\}$.  \ENDIF \ENDFOR
\STATE\label{step:online-end} Output $O$.  \ENDIF \ENDIF
\end{algorithmic}
\end{algorithm}
 
Note that we use an algorithm for
$(\Phi,\Feas)$ where $\Phi$ is a sum-of-values objective in Algorithm
\ref{Algorithm:online-matroid} as a black box. For example if the
underlying feasibility constraint is a partition matroid we execute
the partition-matroid secretary algorithm in steps~\ref{step:redn-1}
and \ref{step:redn-2} as subroutines. Since these subroutines are online algorithms, 
we can execute steps~\ref{step:redn-1} and
\ref{step:redn-2} in parallel with the respective `for' loops in
steps~\ref{step:for-1} and \ref{step:for-2}. This ensures that all
accepted elements have positive profit increment.

\begin{definition}
For a fixed parameter $\eone \leq 1 $, let $\rhob$ be the highest density with $\profit(\OPT(\prefsmall)) \geq (\eone/16)^2 \profit(\OPT)$.
\end{definition}

\begin{lemma}
  \label{lem:main} 
For fixed parameters $\const[1] \geq 1 $, $\const[2] \leq 1$,
$\eone \leq 1$ and $\etwo \leq 1$ suppose that there is no element
with profit more than $\frac{1}{\const[1]}\profit(\OPT)$. Then with
probability at least $\const[2]$, we have that if $\guess$ is the highest density  
such that $\profit(\alg(\prefguess^X)) \geq \frac{\eone}{16}\profit(\alg(X))$,
then $\guess$ satisfies
$\rhob\geq \guess \geq \rhoc$ and
$\profit(\OPT(\prefguess^Y)) \geq \etwo(\eone/16)^2 \profit(\OPT)$.
\end{lemma}

\begin{proof}
We first define two events analogous to the events $E_1$ and $E_2$ in
Section~\ref{sec:unconstrained}.
\begin{align*}
E_1: \profit(\OPT(\preflarge^X)) & \geq \eone
\profit(\OPT(\preflarge)) \\ E_2: \profit(\OPT(\prefsmall^Y)) & \geq
\etwo \profit(\OPT(\prefsmall))
\end{align*}

We claim that the event $E_1$ immediately implies
$\rhob\ge\guess\ge\rhoc$. Furthermore, when $\rhob\ge\guess\ge\rhoc$,
we get the containment
$\prefsmall^Y \subset \prefguess^Y \subset \preflarge^Y$, which
implies
$\profit(\OPT(\prefguess^Y)) \geq \profit(\OPT(\prefsmall^Y))$. This
inequality, when combined with event $E_2$ and the definition of
$\rhob$, proves the second condition. We furthermore claim that $E_1$
and $E_2$ simultaneously hold with probability at least $\const[2]$,
which would give the desired result.

Thus, we are done if we demonstrate that event $E_1$ implies
$\rhob\ge\guess\ge\rhoc$, and that the probability of $E_1$ and $E_2$
occurring simultaneously is at least $\const[2]$; we now proceed to prove each of
these claims in turn.



\begin{claim}
\label{lem:upperbound_guess}
  Event $E_1$ implies that $\rhob \geq \guess$ and so $\prefsmall^Y
  \subset \prefguess^Y$.
\end{claim}

\begin{proof}
  First, by containment and optimality, we observe that
  \[\profit(\OPT(\prefguess))
  \geq \profit(\OPT(\prefguess^X)) \geq \profit(\alg(\prefguess^X)).\] 
  By definition of $\guess$, we have $\profit(\alg(\prefguess^X)) \geq
  \frac{\eone}{16}\profit(\alg(X))$. Furthermore, \[\profit(\alg(X))
  \geq \frac{1}{4}\profit(\OPT(X)) \geq
  \frac{1}{4}\profit(\OPT(\preflarge^X)). \]
  Lemma \ref{lem:goodprefix} gives us $\profit(\OPT(\preflarge)) \geq
  \frac{1}{4}\profit(\OPT)$. This together with event $E_1$ implies
  \[\profit(\OPT(\preflarge^X)) 
  \geq \eone \profit(\OPT(\preflarge)) \geq \frac{\eone}{4}
  \profit(\OPT).\]
  Thus, we have that $\profit(\OPT(\prefguess)) \geq
  \left(\frac{\eone}{16}\right)^2 \profit(\OPT)$.  We have defined
  $\rhob$ to be the largest density for which the previous profit
  inequality holds. Hence we conclude that $\rhob \geq \guess$.
\end{proof}
\begin{claim}
\label{lem:lowerbound_guess}
  Event $E_1$ implies that $\guess \geq \rhoc$.
\end{claim}

\begin{proof}

As stated before, Lemma \ref{lem:goodprefix} gives us
$\profit(\OPT(\preflarge)) \geq \frac{1}{4}\profit(\OPT)$. We have
that $\profit(\alg(\preflarge^X)) \geq \frac{1}{4}
\profit(\OPT(\preflarge^X))$. Now, $E_1$ implies that
  \[\profit(\OPT(\preflarge^X)) 
  \geq \eone \profit(\OPT(\preflarge)) \geq \frac{\eone}{4}
  \profit(\OPT).\] Combining these we get that
  \[\profit(\alg(\preflarge^X)) \geq \frac{\eone}{16} \profit(\OPT) \geq
  \frac{\eone}{16} \profit(\alg(X)).\]

  Since $\guess$ is the largest density for which the above inequality
  holds we have $\guess \geq \rhoc$.
\end{proof}


\begin{claim}
\label{lem:chernoff-matroid}
For a fixed constant $\const[1]$, if no element of $S$ has profit more
than $\frac{1}{\const[1]} \profit( \OPT) $ then $\textrm{Pr}[E_1
  \wedge E_2] \geq 0.52$.
\end{claim}
\begin{proof}
The proof of this claim is similar to that of Lemma
\ref{lemma:chernoff-unconstrained}. We show that $\Pr[E_1 ] \ge 0.76
$ and $\Pr[ E_2] \ge 0.76$; the result then follows by applying the
union bound.  We begin by observing that
$\profit(\OPT(\preflarge^X)) \geq \profit(\OPT(\preflarge) \cap X)$,
and so it suffices to bound the probability that
$\profit(\OPT(\preflarge) \cap X) \geq \eone
\profit(\OPT(\preflarge))$ and likewise the probability that $\profit(\OPT(\prefsmall)
\cap Y) \geq \etwo \profit(\OPT(\prefsmall))$.


We fix an ordering over elements of $\OPT(\preflarge)$, such that the
profit increments are non-increasing. That is, if $L_i$ is the set
containing elements $1$ through $i$ and hence the profit increment of
the $i$th element is $\incprofit(i) := \profit(L_i) -
\profit(L_{i-1})$, then we have $\incprofit(1) \geq \incprofit(2) \geq
\ldots$. Note that such an ordering can be determined by greedily
picking elements from $\OPT(\preflarge)$ such that the profit
increment at each step is maximized.

We set $\const[1] \geq \lOneMinK$, now the fact that no element in $S$ has
profit more than $\frac{1}{\lOneMinK} \profit( \OPT) $ implies no element by
itself has profit more than $1/\lOneSize \ \profit(\OPT(\preflarge))$, since
$ \profit(\OPT(\preflarge)) \geq 1/3 \profit(\OPT)$. Profit increments
of elements are upper bounded by their profits, therefore we can apply
Lemma \ref{lemma:concentration} with $\OPT(\preflarge)$ as the fixed
set and profit increments as the weights. By optimality of
$\OPT(\preflarge)$ we have that these profit increments are
non-negative hence the required conditions of Lemma
\ref{lemma:concentration} hold and we have $ \profit(\OPT(\preflarge)
\cap X) \geq \eone \profit( \OPT(\preflarge)) $ with probability at
least $0.76$.  Since $\profit(\OPT(\preflarge^X)) \geq
\profit(\OPT(\preflarge) \cap X)$ we get $\Pr[ E_1 ] \ge 0.76$ with
$\eone = \lOneRatio$.

By definition of $\rhob$ the profit of $\OPT(\prefsmall)$ is at least
$\left(\frac{\eone}{9}\right)^2 \profit(\OPT)$. With $\const[1] \geq
\lOneSize \left( \frac{9}{\eone}\right)^2 $ and $\etwo = \lOneRatio$ we can again
apply Lemma \ref{lemma:concentration} to show that $\Pr[ E_2]
\ge0.76$. Hence we can set $\const[2] = 0.52$ and this completes the proof
of the claim.
\end{proof}

With the demonstration that the above three claims hold, our proof is complete.
\end{proof}

To conclude the analysis, we show that if the online algorithms for
$(\sdf_\tau,\Feas)$ and $(s,\Feas)$ have a
competitive ratio of $\alpha$, then we obtain an $O(\alpha)$
approximation to $\profit(\OPT(P^Y_\tau))$. We therefore get the
following theorem.

\begin{theorem}
\label{thm:online-matroid}
If there exists an $\alpha$-competitive algorithm for the matroid
secretary problem $(\Phi, \Feas)$ where $\Phi$ is a sum-of-values objective,
then Algorithm~\ref{Algorithm:online-matroid} 
achieves a competitive ratio of
$O(\alpha)$ for the problem $(\profit, \Feas)$.
\end{theorem}

Before we proceed to prove Theorem~\ref{thm:online-matroid}, we show
that in steps \ref{step:online-start} to \ref{step:online-end} the
algorithm obtains a good approximation to $\OPT(\prefguess^Y)$.

\begin{lemma}
  \label{lem:secondstage}
Suppose that there is an $\alpha$-competitive algorithm for
$(\Phi,\Feas)$ where $\Phi$ is any sum-of-values objective. For a fixed set
$Y$ and threshold $\guess$, satisfying $\guess \geq \rhoc$, we have $
\E_\sigma[\profit(O_1) + \profit(O_2)] \geq \alpha
\ \profit(\OPT(\prefguess^Y))$, where the expectation is over all
permutations $\sigma$ of $Y$.
\end{lemma}

\begin{proof}
 The threshold $\tau$ is either equal to or strictly greater than
 $\rhoc$. In the former case $e_\rhoc$ must have been in the sample
 set $X$ and hence $O_1, O_2 \subseteq \bP_\rhoc$. By Lemma \ref{lem:rhoc-indep} we show that $O_1$ and $O_2$ are
 $\rhoc$-bounded and hence $\tau$-bounded. On the other hand if $\tau
 > \rhoc$ we can again apply Lemma \ref{lem:rhoc-indep} and get that
 $O_1$ and $O_2$ are $\tau$-bounded.
 
Hence by Proposition \ref{lem:profit-lb} we get the inequalities
$\E_\sigma[\profit(O_1)] \geq \E_\sigma\left[ \sdf_\guess(O_1)\right]$
and $\E_\sigma[\profit(O_2)] \geq \E_\sigma[ \fdf_\guess(O_2)]$.
  
  By applying the $\alpha$-competitive matroid secretary algorithm with
  objective $\sdf_\guess$ (Step 11 of Algorithm
  \ref{Algorithm:online-matroid} ) we get
  \begin{align*}
    \E_\sigma[ \sdf_\guess(O_1) ] &\geq \alpha \times \sdf_\guess(
    H_\guess^Y) \\ &\geq \alpha \times \sdf_\guess(\OPT(\prefguess^Y)),
  \end{align*}
\noindent where the second inequality follows from the optimality of $H_\guess^Y$. 

Next we bound $\E_\sigma[\fdf_\guess(O_2)]$. Let
$K$ be the largest feasible subset contained in $\prefguess^Y$. The
fact that the underlying algorithm is $\alpha$-competitive implies
$\E_\sigma[s(O_2)] \geq \alpha \times s(K)$.

Note that, as observed above, $\prefguess^Y \subset \bP_\rhoc$. Since $K \subset \prefguess^Y $, by definition of $\rhoc$ we get that $s(K) \leq
s^*$. So, for $O_2$ we have
  \begin{align*}
    \E_\sigma[\fdf_\guess(O_2)] &\geq \E_\sigma[\guess s(O_2) -
      \cost(s(O_2)]\\ &\geq \E_\sigma\left[\guess
      \left(\frac{s(O_2)}{s(K)}\right) s(K) -
      \left(\frac{s(O_2)}{s(K)}\right) \cost(s(K))\right] \\
      & = \E_\sigma\left[\frac{s(O_2)}{s(K)}\right] \left(\guess s(K) -
    \cost(s(K))\right)\\ 
    &\geq \alpha (\guess s(K) - \cost(s(K)))\\ & = \alpha
    \ \fdf_\guess(K) \\ & \geq \alpha \ \fdf_\guess(\OPT(\prefguess^Y)).
  \end{align*}

\noindent Since $s(\OPT(\prefguess^Y)) \leq s(K) \leq s^* \leq
\dcostinv(\guess)$ we get the last inequality by applying Lemma
\ref{lem:monotone}.

  The conclusion of the lemma now follows from the decomposition
  $\profit(\OPT(\prefguess^Y)) = \sdf_\guess(\OPT(\prefguess^Y)) +
  \fdf_\guess (\OPT(\prefguess^Y) ) $.
  \end{proof}

\begin{proof}[Proof of Theorem~\ref{thm:online-matroid}]
With probability $\frac{1}{2}$ we apply the standard secretary algorithm which is 
$e$-competitive. If an element has profit more than
$\frac{1}{\const[1]} \profit(\OPT)$, in expectation we get a profit of
$\frac{1}{2 \const[1] e}$ times the optimal.

We have $\Pr[E_1 \wedge E_2]\geq k_2 $, for a fixed constant $k_2$. Also, the events $E_1$ and $E_2$ 
depend only on what elements are in $X$ and $Y$, 
and not on their ordering in the stream. So conditioned on $E_1$
and $E_2$, the remaining stream is still a uniformly random
permutation of $Y$. Therefore, if no element has profit more than
$\frac{1}{\const[1]} \profit(\OPT)$ we can apply the second inequality
of Lemma \ref{lem:main} and Lemma \ref{lem:secondstage} along with the
fact that we output $O_1$ and $O_2$ with probability $\frac{1}{4}$
each,  to show that
\begin{align*}
\E[\profit(O)] &\geq \frac{1}{4} \E[\profit(O_1) +
  \profit(O_2)\ |\ E_1 \wedge E_2]\times \Pr[E_1 \wedge E_2]  \\ &\geq \frac{\alpha \const[2]
}{4}\E[\profit(\OPT(\prefguess^Y))\ |\ E_1 \wedge E_2]\\ &\geq
\frac{\alpha \const[2] \etwo}{4}\left(\frac{\eone}{16}\right)^2
\profit(\OPT).
\end{align*}

Overall, we have that \[\E[\profit(O)] \geq \min\left\{\frac{ \alpha
  \const[2]\etwo}{4}\left(\frac{\eone}{16}\right)^2, \frac{1}{2
  \const[1] e}\right\} \cdot \profit(\OPT). \] Since all the involved
parameters are fixed constants we get the desired result.
\end{proof}

\section{Multi-dimensional Profit Maximization}
\label{sec:multi}
In this section, we consider the GSP with a multi-dimensional profit
objective. Recall that in this setting each element $e$ has $\ell$
different sizes $s_1(e), \ldots, s_\ell(e)$, and the cost of a subset is
defined by $\ell$ different convex functions $\cost_1, \ldots,
\cost_\ell$. The profit function is defined as $\profit(A) = v(A) -
\sum_i \cost_i(s_i(A))$.

As in the single-dimensional setting, we partition $U$ into two sets
$\I$ and $\F$ with $\F=\{e\in U : \argmax_\alpha \profit(\alpha
e)<1\}$. We first claim that, as before, an optimal solution cannot
contain too many elements of $\F$. 
\begin{lemma}
  \label{lem:multi-costs:F}
  We have that
  $|\OPT \cap \F| \leq \ell$.  
\end{lemma}

\begin{proof}
  Suppose, towards a contradiction, that $|\OPT \cap \F| \geq \ell+1$.
  For $i \in \{1, \ldots, \ell\}$, let $m_i$ be any element in $\OPT$
  with $s_i(m_i) = \max_{e\in\OPT} s_i(e)$. Since $|\OPT \cap \F|
  \geq \ell+1$, there exists $o \in (\OPT \cap \F) \setminus \{m_1,
  \ldots, m_\ell\}$ such that $s_i(o) \leq s_i(\OPT \setminus
  \{o\}) 
  $ for all $i$. This implies that when we compare the marginal cost
  of adding another copy of $o$ to $\{o\}$ against the marginal cost
  of adding $o$ to $\OPT \setminus \{o\}$, we have by convexity \[\sum_i
  \cost_i(s_i(o + o)) - \cost_i(s_i(o)) \leq\sum_i \cost_i(s_i(\OPT))
  - \cost_i(s_i(\OPT \setminus \{o\})). \]
  Therefore, we have $\profit(\OPT) - \profit(\OPT
  \setminus \{o\}) \leq \profit(o + o) - \profit(o) < 0$ since $o \in
  \F$, and this contradicts the optimality of $\OPT$.
\end{proof}

\noindent We therefore focus on approximating $\profit$ over $\I$
and devote this section to the unconstrained problem $(\profit,
2^U)$. 

The main challenge of this setting is that we cannot summarize the
value-size tradeoff that an element provides by a single density
because the element can be quite large in one dimension and very small
in another. Our high level approach is to distribute the value of each
element across the $\ell$ dimensions, thereby defining densities and
decomposing profit across dimensions appropriately. We do this in such
a way that a maximizer of the $i$th dimensional profit for some
dimension $i$ gives us a good overall solution (albeit at a cost of a
factor of $\ell$).

Formally, let $\densityv : U \rightarrow \mathbb{R}^\ell$ denote an
$\ell$-dimensional vector function $\densityv(e) = (\rho_1(e), \ldots,
\rho_\ell(e))$ that satisfies $\sum_i \densityv_i(e)s_i(e) =
v(e)$ for all $e$. We set $v_i(e)=\densityv_i(e)s_i(e)$ and
$\profit_i(A) = v_i(A)-\cost_i(s_i(A))$ and note that $\profit(A) =
\sum_i \profit_i(A)$. Let $\FOPT_i$ denote the maximizer of
$\profit_i$ over $\I$. Then, $\profit(\FOPT) \leq \sum_i
\profit_i(\FOPT_i)$.

Given this observation, it is natural to try to obtain an
approximation to $\profit$ by solving for $\FOPT_i$ for all $i$ and
rounding the best one. This does not immediately work: even if
$\profit_i(\FOPT_i)$ is very large, $\profit(\FOPT_i)$ could be
negative because of the profit of the set being negative in other
dimensions. We will now describe an approach for defining and finding
density vectors such that the best set $\FOPT_i$ indeed gives an
$O(\ell)$ approximation to $\OPT(\I)$. We first define a quantity
$\ideal_i(\gamma_i)$ which bounds the $i$th dimensional profit that can
be obtained by any set with elements of $i$th density at most
$\gamma_i$: $\ideal_i(\gamma_i) = \max_t (\gamma_i t - \cost_i(t))$. We can
bound $\profit_i(\FOPT_i)$ using $\ideal_i(\cdot)$.
\begin{lemma} 
  \label{lem:multi-cost:tail-bd}
  For a given density $\gamma_j$, let $A=\{a \in \FOPT_j : \rho_j(a)
  \geq \gamma_j\}$. Then $\profit_j(\FOPT_j) \leq \profit_j(A) +
  \ideal_j(\gamma_j)$.
\end{lemma}
\begin{proof}
  Subadditivity implies $\profit_j(\FOPT_j) \leq \profit_j(A) +
  \profit_j(\FOPT_j \setminus A)$. Since the elements $e$ in $\FOPT_j
  \setminus A$ have $\rho_j(e) \leq \gamma_j$, we have
  $\profit_j(\FOPT_j \setminus A) \leq \max_t \gamma_j t -
  \cost_j(t) = \ideal_j(\gamma_j)$.
\end{proof}

In order to obtain a uniform bound on the profits
$\profit_j(\FOPT_j)$, we restrict density vectors as follows.  We call
a vector $\densityv(e)$ {\em proper} if it satisfies the following
properties:
\begin{enumerate}
\item[(P1)] $\sum_i \densityv_i(e)s_i(e) = v(e)$
\item[(P2)] $\ideal_i(\densityv_i(e)) = \ideal_j(\densityv_j(e))$ for
  all $i,j\in[1,\ell]$; we denote this quantity by $\ideal(e)$.
\end{enumerate}
The following lemma is proved in Section~\ref{app:computing}.
\begin{lemma}
For every element $e$, a unique proper density vector exists and can
be found in polynomial time. 
\end{lemma}
Finally, we note that proper density vectors induce a single ordering
over elements. In particular, since the $\ideal_i$s are monotone,
$\rho_i(e) \geq \rho_i(e')$ if and only if $\ideal(e) \geq
\ideal(e')$.  We order the elements $e_1, \ldots, e_n$ in decreasing
order of $\ideal$. Note that each $\FOPT_i$ is a (fractional) prefix of
this sequence. Let $\FOPT_1$ be the shortest prefix. Let $\alg =
\{e_1, \ldots, e_{k_1}\}$ denote the integral part of $\FOPT_1$ and
$\ebar = e_{k_1 + 1}$(\textit{i.e.}, $\FOPT_1$'s unique fractional element if it
exists). 

First, we need the following fact about $\FOPT_1$. It implies that the
multidimensional profit function $\pi$ is monotone when restricted to
subsets of $\FOPT_1$.
\begin{lemma}
  \label{lem:multi-costs:monotone}
  Consider subsets $A_1, A_2 \subset \FOPT_1$. If $A_1 \supset A_2$,
  then $\profit(A_1) \geq \profit(A_2)$. Furthermore, $\profit(A) \geq
  \profit_i(A)$ for all $i$ if $A \subset \FOPT_1$.
\end{lemma}

\begin{proof}
  For any $i$, we have $\FOPT_1 \subset \FOPT_i$. Since $\FOPT_i$ is
  the fractional prefix that optimizes $\profit_i(\cdot)$, applying
  Lemma~\ref{lemma:monotone-integral} to $\profit_i(\cdot)$ implies
  that $\profit_i(A_1) \geq \profit_i(A_2) \geq 0$. By summing, we
  have $\profit(A_1) \geq \profit(A_2)$, and also
  that $\profit(A_1) \geq \profit_i(A_1)$. Setting $A_1$ as $A$ gives
  us the second claim.
\end{proof}

We get the following lemma by noting that $\profit_1(\FOPT_1)
\geq \ideal_1(e')$.
\begin{lemma}
  \label{lem:multi-costs:I}
  For proper $\densityv(e)$s and $\alg$ and $\ebar$ as defined above,
  for every $i$, $\profit_i(\FOPT_i)\le
  \profit_i(\alg)+\profit_1(\FOPT_1) \le
  2\profit(\FOPT_1)$. Furthermore, $\profit(\FOPT) \leq
  \ell(\profit(\alg) + \profit(\ebar))$.
\end{lemma}
\begin{proof}
  We begin by proving that $\profit_1(\FOPT_1) \geq \ideal_1(e')$. Let
  $\rho'_i = \min_{a \in \FOPT_i} \rho_i(a)$.  We observe that
  $\rho_i(e') \leq \dcost_i(s_i(\FOPT_i))$ since otherwise, we could
  have included an additional fractional amount of $e'$ to $\FOPT_i$
  and increased its profit. Together with Claim~\ref{lem:opt_frac}, we
  have $\dcostinv_i(e') \leq s_i(\FOPT_i) \leq
  \dcostinv_i(\rho'_i)$. Thus, applying Claim~\ref{lem:monotone} and
  the fact that $\rho'_i\geq \rho_i(e')$, we
  have 
  \begin{align*}   
    \ideal_i(e')
    = \rho_i(e')\dcostinv_i(e') - \cost_i(\dcostinv_i(e'))
    \leq \rho'_i\dcostinv_i(e') - \cost_i(\dcostinv_i(e'))
    \leq \rho'_is_i(\FOPT_i) - \cost_i(s_i(\FOPT_i)).
  \end{align*} 
  Since $\FOPT_i$ is $\rho'_i$-bounded in the $i$th dimension, we have
  $\profit_i(\FOPT_i) \geq \rho'_is_i(\FOPT_i) - \cost_i(s_i(\FOPT_i))$
  by Proposition~\ref{lem:profit-lb}. This proves that $\profit_1(\FOPT_1)
  \geq \ideal_1(e')$.

  Since $\densityv$ is proper, we have $\ideal_i(\ebar) =
  \ideal_1(\ebar)$. Recall that $\ebar$ is not in $\alg$, the integral
  subset of $\FOPT_1$, so we have 
  $\ideal_i(\ebar) = \ideal_1(\ebar) \leq
  \profit_1(\FOPT_1)$. Together with
  Lemma~\ref{lem:multi-cost:tail-bd}, this gives us  
  \begin{equation}
    \label{eq:multi-costs:I}
    \profit_i(\FOPT_i) \leq \profit_i(\alg) + \profit_1(\FOPT_1).
  \end{equation}
  Now, Lemma~\ref{lem:multi-costs:monotone} gives us $\profit(\FOPT_1)
  \geq \profit_i(\FOPT_1) \geq \profit_i(\alg)$ and $\profit(\FOPT_1)
  \geq \profit_1(\FOPT_1)$, so we have $\profit_i(\FOPT_1) \leq
  2\profit(\FOPT_1)$ as claimed.


  Summing Equation~\eqref{eq:multi-costs:I} over $i$ and applying
  subadditivity to $\profit_1(\FOPT_1)$, we have
  \begin{align*}       
    \sum_i \profit_i(\FOPT_i) &\leq \sum_{i \neq 1}
    \left[\profit_i(\alg) + \profit_1(\FOPT_{1})\right]
    + \profit_1(\FOPT_{1}) \\
    &\leq \sum_{i \neq 1} \profit_i(\alg) + \ell(\profit_1(\alg) +
    \profit_1(\alpha\ebar))\\
    &= \profit(\alg) + (\ell-1)\profit_1(\alg) + \ell\profit_1(\alpha\ebar).
  \end{align*}
  Applying Lemma~\ref{lem:multi-costs:monotone} 
  we have $\profit(\alg) \geq \profit_1(\alg)$ and
  $\profit(\alpha\ebar) \geq \profit_1(\alpha\ebar)$, and so
  \[\profit(\alg) + (\ell-1)\profit_1(\alg) + \ell\profit_1(\alpha\ebar) \leq
  \ell(\profit(\alg) + \profit(\alpha\ebar)).\]
  The conclusion now follows from the fact we are considering only
  elements from $\I$ and
  $\profit(\FOPT) = \sum_i \profit_i(\FOPT) \leq \sum_i
  \profit_i(\FOPT_i)$.
\end{proof}
 
Lemmas~\ref{lem:multi-costs:F} and \ref{lem:multi-costs:I} together
give $\profit(\OPT)\leq \ell(\profit(\alg)+2\max_e\profit(e))$, and
therefore imply an offline $3\ell$-approximation for $(\profit,2^U)$ in
the multi-dimensional setting.


\paragraph{The online setting.} Note that proper densities essentially
define a $1$-dimensional manifold in $\ell$-dimensional space. We can
therefore hope to apply our online algorithm from
Section~\ref{sec:unconstrained} to this setting. However, there is a
caveat: the algorithm from Section~\ref{sec:unconstrained} uses the
offline algorithm as a subroutine on the sample $X$ to estimate the
threshold $\guess$; na\"{i}vely replacing the subroutine by the
$O(\ell)$ approximation described above leads to an $O(\ell^2)$ competitive
online algorithm\footnote{Note the $(1-1/k_1)^2$ factor in the final
  competitive ratio in Theorem~\ref{theorem:weighted}; this factor is
  due to the use of the offline subroutine in determining
  $\guess$.}. In order to improve the competitive ratio to $O(\ell)$ we
need to pick the threshold $\guess$ more carefully.

\begin{algorithm}
\caption{Online algorithm for multi-dimensional $(\profit,2^U)$}
\label{Algorithm:multi-costs:weighted}
\begin{algorithmic}[1]
\STATE  With probability $1/2$ run the classic secretary algorithm
to pick the single most profitable element else execute the following steps.
\STATE Draw $k$ from $\textrm{Binomial}(n,1/2)$.
\STATE Select the first $k$ elements to be in the sample $X$. Unconditionally reject these elements.
\STATE Let $\guess$ be largest density such that $e_\guess \in X$
satisfies $\profit(\prefguessbar^X) + \profit(e_\guess) 
\geq \frac{\beta}{2} \profit_i(\FOPT_i(X))$ for all $i$, for a constant $\eone$.
\STATE Initialize selected set $O \leftarrow \emptyset$.
\FOR {$i \in Y = U \setminus X$}
   \IF { $\profit(O \cup \{i\}) - \profit(O) \geq 0$ and $\rho(i) \geq \guess$ and $i \notin \F$}
           \STATE $O \leftarrow O \cup \{ i \}$
    \ELSE 
         \STATE Exit loop.
    \ENDIF
\ENDFOR
\end{algorithmic}
\end{algorithm}

We define $\guess$ to be the largest density with $e_\guess \in X$
such that for an appropriate constant $\eone$,
$\profit(\prefguessbar^X) + \profit(e_\guess) \geq
\frac{\beta}{2} \profit_i(\FOPT_i(X))$ for all $i$.
For a set $T$, let $\FOPT_i(T)$ denote the maximizer of $\profit_i$
over $T \cap\I$ and let $\FP(T) = \cap_i\FOPT_i(T)$ denote the
shortest of these prefixes. Recall that
$\FP(U)=\FOPT_1$. 
Let $\rhoc$ denote the smallest density in $\FOPT_1$. That
is, $\FOPT_1=\FP(P_\rhoc)$. Our analysis relies on the following two
events:
\begin{align*}
  E_1: \profit(\FP(\preflarge^X)) \geq \eone
  \profit(\FP(\preflarge)), \quad
  E_1': \profit(\OPT(X)) \geq \eonep \profit(\OPT).
 \end{align*}
$E_1$ implies the following sequence of inequalities; here the second
inequality follows from Lemma~\ref{lem:multi-costs:I}.
\begin{align}
  \label{ineq:multi-costs:rhoc}
\profit(\FP(\preflarge^X))
\ge \eone\profit(\FOPT_1) 
\ge \frac{\eone}{2}\profit_i(\FOPT_i)
\ge \frac{\eone}{2}\profit_i(\FOPT_i(X)).
\end{align}
This implies $\guess\ge\rhoc$. 
Formally, we have the following claim.
\begin{lemma}
  \label{lem:multi-costs:rhoc}
  Conditioned on $E_1$, we have $\guess \geq \rhoc$.
\end{lemma}

\begin{proof}
  Since $\preflargebar^X \subset \preflargebar = \alg$,
  Lemma~\ref{lem:multi-costs:monotone} implies that $\preflargebar^X
  \subset \FP(\preflarge^X)$. 
  Event $E_1$ guarantees that $\preflarge^X$ is non-empty, so let
  $\rho'$ be the minimum density of $\preflarge^X$. We have
  $P_{\rho'}^X = \preflarge^X$, which implies $\bar{P}_{\rho'}^X
  \subset \preflargebar^X \subset \FP(\preflarge^X)$ and $\preflarge^X
  = P_{\rho'}^X = \bar{P}_{\rho'}^X \cup \{e_{\rho'}\}$ by
  definition. So, we can write $\FP(\preflarge^X) = \FP(P_{\rho'}^X) =
  \bar{P}_{\rho'}^X \cup \{\alpha' e_{\rho'}\}$ for some $\alpha' \in
  [0,1]$. Using subadditivity and the fact that $e_{\rho'} \in \I$
  gives us
  \begin{equation}
    \label{eq:multi-costs:sample-lb}
    \profit(\bar{P}_{\rho'}^X) + \profit(e_{\rho'}) \geq
    \profit(\bar{P}_{\rho'}^X \cup \{\alpha' e_{\rho'}\}) =
    \profit(\FP(\preflarge^X)).  
  \end{equation}
  
  Now, threshold $\guess$ as selected by Algorithm
  \ref{Algorithm:multi-costs:weighted} is the largest density such
  that $\profit(\prefguessbar^X) + \profit(e_\guess) \geq
  \frac{\eone}{2} \profit(\FOPT_i(X))$ for all $i$. Since step 4 of the
  algorithm would have considered $\profit(\bar{P}_{\rho'}^X) +
  \profit(e_{\rho'})$, Equations~\eqref{eq:multi-costs:sample-lb} and
  \eqref{ineq:multi-costs:rhoc} imply $\guess \geq \rho'$. By definition,
  $\rho' \geq \rhoc$, so we conclude that $E_1$ implies $\guess \geq \rhoc$.
\end{proof}

Furthermore, by definition, $e_\guess \in X$ and so $e_\guess \notin
\prefguess^Y$, which implies that $\prefguess^Y \subset
\prefguessbar$. Since $\guess \geq \rhoc$ implies $\prefguessbar
\subset \preflargebar = \alg$, Lemma~\ref{lem:multi-costs:rhoc}
gives us the following lemma.
\begin{lemma}
  \label{lem:multi-costs:guess-lb}
  Conditioned on $E_1$, we have $O = \prefguess^Y \subset
  \prefguessbar \subset \alg$.
\end{lemma}

Summing over all dimensions and applying event $E_1'$ gives us
\begin{align*}
  \ell(\profit(\prefguessbar^X) + \profit(e_\guess))
  \geq \frac{\beta}{2} \sum_i \profit_i(\FOPT_i(X))
  \geq \frac{\beta}{2} \sum_i \profit_i(\OPT(X))
  = \frac{\beta}{2} \profit(\OPT(X)) \ge \frac{\eone\eonep}{2}\profit(\OPT).
\end{align*}
So if we define $\rhob$ to be the highest density such that
$\profit(\prefsmallbar) + \profit(e_\rhob) \geq \frac{\eone\eonep}{2\ell}
\profit(\OPT)$, then we get $\rhob\ge\guess$. Then, as before we can
define the event $E_2$ in terms of $\rhob$ to conclude that
$\profit(P_\tau^Y)$ is large enough. Formally, we define $E_2$ for
some fixed constant $\etwo$ as
follows.
\begin{align*}
  &E_2: \profit(\prefsmall^Y) \geq \etwo
  \profit(\prefsmall).
\end{align*}
Since $\prefsmall^Y \subset \prefguess^Y \subset \prefguessbar \subset
\alg$, 
applying Claim~\ref{lem:multi-costs:monotone}, we have $\profit(O) =
\profit(\prefguess^Y) \geq
\profit(\prefsmall^Y)$. Therefore, conditioning on events $E_1$,
$E_1'$ and $E_2$, we get
\begin{align*}
\profit(O) \ge \profit(\prefsmall^Y)  \ge \frac{\eone\eonep\etwo}{2\ell} \profit(\OPT).
\end{align*}

To wrap up the analysis, we argue that the probability of these events
is bounded from below by a constant. Using
Lemma~\ref{lemma:concentration} as in the proof of
Lemma~\ref{lemma:chernoff-unconstrained} , if no element has profit at
least $\frac{1}{\ell\const[3]}\profit(\OPT)$, then events $E_1$, $E_1'$
and $E_2$ each occur with probability at least
$0.76$. 
Using a union bound, we have that $\Pr[E_1 \wedge E_1' \wedge E_2]
\geq 0.28$. This proves the following
lemma.

\begin{lemma}
  \label{lem:multi-costs:eventprobs}
Suppose that $\max_e\profit(e)\le (1/\const[3]\ell)\profit(\OPT)$. Then,
$\Pr[E_1 \wedge E_1' \wedge E_2] \geq 0.28$.
\end{lemma}
Via a similar argument as for Theorem~\ref{theorem:weighted}, we get
\begin{theorem}
  \label{theorem:multi-costs:weighted} 
  Algorithm~\ref{Algorithm:multi-costs:weighted} is
  $O(\ell)$ competitive for $(\profit,2^U)$ where $\profit$ is a
  multi-dimensional profit function.
\end{theorem}

\subsection{Computing Proper Densities}
\label{app:computing}
In this subsection, we show that proper $\densityv$'s exist and can be
efficiently computed.  To do this, we make the following assumptions
about the marginal cost functions:
\begin{enumerate}
\item[(A1)] The marginal cost functions $\dcost_j(\cdot)$ are unbounded.
\item[(A2)] They satisfy $\dcost_j(s_j(U)) > \max_{e \in U} v(e)$ for
  all $i$.
\end{enumerate}
Note that these assumptions do not affect $C_i(t)$ for sizes $t \leq
s_i(U)$, and so have no effect on either the profit function or the
optimal solution as well.  We first prove that there exists a proper
$\densityv(e)$ for each element $e$. Observe that properties (P1) and (P2)
together uniquely define $\rho(e)$. Therefore, we only need to find $x^*$
satisfying the equation
\begin{equation}
  \label{eq:multi-costs:densities}
  \sum_j \ideal_j^{-1}(x^*) s_j(e) = v(e);
\end{equation}
then our proper density $\rho(e)$ is given by $\rho_j(e) =
\invideal_j(x^*)$. By Assumption (A1), $\ideal_j$ is a strictly-increasing
continuous function with $\ideal_j(0) =0$ and $\lim_{\gamma
  \rightarrow \infty} \ideal_j(\gamma) = \infty$, so its inverse $\invideal_j$
is well-defined and is also a strictly-increasing continuous function
with $\invideal_j(0) =0$ and $\lim_{x
  \rightarrow \infty} \invideal_j(x) = \infty$. Thus, the solution $x^*$ exists.

Next, we show that for any element $e$, we can efficiently compute
$\densityv(e)$. In the following, we fix an element $e$, and use $s_j$ and
$v$ to denote $s_j(e)$ and $v(e)$,
respectively. 
We define $I_j(t) = \ideal_j(\dcost_j(t))$. Let $x^*$ be the solution
to Equation \ref{eq:multi-costs:densities} for the fixed element $e$
and note that $x^* = \ideal(e)$.


In the following two lemma statements and proofs, we focus on a single
dimension $j$ and remove subscripts for ease of notation.
First, we prove that we can easily compute $I(t)$.
\begin{lemma}
  \label{lem:multi-costs:I(t)}
  We have $I(t) = \dcost(t)t - \cost(t)$.
\end{lemma}

\begin{proof}
  Let $t' = \dcostinv(\dcost(t))$. Since $t'$ is the maximum size $r
  \geq t$ such that $\dcost(t) \geq \dcost(r)$, by monotonicity, we
  have that $\dcost(\cdot)$ is constant in $[t, t']$. This implies $\cost(t') -
  \cost(t) = \dcost(t)(t' - t)$ and so,
  \[\dcost(t)t' - \cost(t') = \dcost(t)t- \cost(t).\]
  The lemma now follows from the fact that $I(t) = \ideal(\dcost(t))$
  is the LHS of the above equation.
\end{proof}
\noindent
Next, we prove a lemma that helps us determine $\invideal(x)$ given
$x$.
\begin{lemma}
  \label{lem:multi-costs:find-density}
  Given $x$ and positive integer $t$ such that 
  $I(t) \leq x < I(t+1)$, 
  we have 
  \begin{equation*}
    \label{eq:multi-costs:find-density}
    \invideal(x) = 
    \frac{x + \cost(t)}{t}.
  \end{equation*} 
\end{lemma}
\begin{proof}
  By definition of $I(\cdot)$, we have 
  $\ideal(\dcost(t)) \leq x < \ideal(\dcost(t+1))$, and since
  $\ideal(\cdot)$ is strictly increasing, we get \[\dcost(t) \leq
  \invideal(x) < \dcost(t+1).\] By definition of $\dcostinv(\cdot)$,
  this gives us $t \leq \dcostinv(\invideal(x)) < t+1$, and therefore
  $\dcostinv(\invideal(x)) = t$, since $\dcost(\cdot)$ changes only on the
  integer points. 

  Thus, we have 
  \[x = \ideal(\invideal(x)) = \invideal(x) t -\cost(t)\]
  and solving this for $\invideal(x)$ gives us the desired
  equality. 
\end{proof}


Lemma~\ref{lem:multi-costs:find-density} leads to the
\textsc{Find-Density} algorithm which, given a profit $x$, uses a
binary search to compute $\invideal_j(x)$.  Together with
Lemma~\ref{lem:multi-costs:find-ideal}, this enables us to
determine $x^*$ by first using binary search, and then solving linear
equations.
\begin{algorithm}{$\textsc{Find-Density}(x, s)$}
  \caption{Given $x$ and sizes $s_j$, find $\invideal(x)$ and $t_1,
    \ldots, t_\ell$
    satisfying $I_j(t_j) \leq x < I_j(t_j+1)$.}
\begin{algorithmic}[1]
\FOR {$j = 1$ to $\ell$}
\STATE Binary search to find integral $t_j \in [0, s_j(U))$ satisfying
$I_j(t_j) \leq x < I_j(t_j + 1)$.
\STATE Set $\rho_j \leftarrow (x + \cost_j(t_j))/t_j$.
\ENDFOR
\RETURN $(\rho, t)$, where $\rho$ is the vector $(\rho_1, \ldots,
\rho_\ell)$ and $t$ is the vector $(t_1, \ldots, t_\ell)$.
\end{algorithmic}
\end{algorithm}

\begin{lemma}
  \label{lem:multi-costs:find-ideal}
  Suppose we have a positive integer $x$ such that 
  \begin{equation}
    \label{eq:multi-costs:find-ideal}
    \sum_j \invideal_j(x) s_j \leq v < \sum_j \invideal_j(x+1) s_j,
  \end{equation}
  and positive integers $t_1, \ldots, t_\ell$ such that $I_j(t_j) \leq x
  < I_j(t_j + 1)$ for all $j$.  Then the solution $x^*$ to
  Equation~\eqref{eq:multi-costs:densities} 
  is precisely:
  \[x^* = \frac{v - \sum_j \cost_j(t_j)(s_j/t_j)}{\sum_j(s_j/t_j)}.\]
\end{lemma}
\begin{proof}
  Equation~\eqref{eq:multi-costs:find-ideal} and the monotonicity of
  the $\invideal_j$'s imply that $x \leq x^* < x+1$ and so, we have
  $I_j(t_j) \leq x^* < x+1 \leq I_j(t_j + 1)$ for all $j$. Applying
  Lemma~\ref{lem:multi-costs:find-density} to dimension $j$ gives us
  $\invideal_j(x^*) = (x^* -
  \cost_j(t_j))/t_j$. 
  Then, applying Equation~\eqref{eq:multi-costs:densities}, we get
  \[v = \sum_j \invideal_j(x^*)s_j = \sum_j \left[
    \frac{x^* + \cost_j(t_j)}{t_j}\right]s_j.\]
  Solving this for $x^*$ gives the claimed equality.
\end{proof}


\begin{algorithm}{$\textsc{Find-Proper-Density}(v,s)$}
  \caption{Given sizes $s_j$ and value $v$, find $\rho$ satisfying (P1)
    and (P2).
  }
\begin{algorithmic}[1]
\STATE Binary search to find integral $x \in [0, \min_j I_j(s_j(U)))$
satisfying
\[\sum_j \delta^-_j s_j \leq v < \sum_j \delta^+_j s_j,\]
where $(\delta^-, t) =
\textsc{Find-Density}(x, s)$ and $(\delta^+, t') =
\textsc{Find-Density}(x+1, s)$.
\STATE Set \[x^* \leftarrow \frac{v -
  \sum_j\cost_j(t_j)(s_j/t_j)}{\sum_j(s_j/t_j)}.\]
\FOR {$j = 1$ to $\ell$}
\STATE Set $\rho_j \leftarrow (x^* + \cost_j(t_j))/t_j$.
\ENDFOR
\RETURN $\rho$
\end{algorithmic}
\end{algorithm}

We need the following lemma to show that the binary search upper
bounds in both algorithms are correct.
\begin{lemma}
  \label{lem:multi-costs:bin-upper-bd}
  For proper density $\rho$, we have 
  $x^* < I_j(s_j(U))$ for all dimensions $j$.
\end{lemma}

\begin{proof}
  If an element has zero size in dimension $j$, then we can ignore
  $\cost_j(\cdot)$. So without loss of generality, we assume that $s_j > 0$
  for all $j$. Property (P1) gives us $v = \sum_i \rho_i s_i \geq
  \rho_j s_j$ and so $\rho_j \leq v/s_j \leq v$.  From
  assumption (A2), we have that $\dcost_j(s_j(U)) > v \geq
  \rho_j$. This implies that 
  \begin{align*}
   \ideal_j(\rho_j) 
   < \ideal(\dcost_j(s_j(U))) 
   &= I_j(s_j(U))
 \end{align*}
 Since $\rho$ is proper, we have $\ideal_j(\rho_j) = x^*$ and this
 proves the lemma.
\end{proof}

\begin{theorem}
  Algorithms \textsc{Find-Density} and \textsc{Find-Proper-Density} run in
  polynomial time and are correct.
\end{theorem}
\begin{proof}
  Lemmas \ref{lem:multi-costs:find-density} and
  \ref{lem:multi-costs:find-ideal} show that given correctness of the
  binary search upper bounds, the output is
  correct. Lemma~\ref{lem:multi-costs:bin-upper-bd} implies that the
  binary search upper bound of \textsc{Find-Proper-Density} is
  correct. We observe that \textsc{Find-Density} is only invoked for
  integral profits $x < I_j(s_j(U))$. Therefore, we have $I_j(t') \leq
  x < I_j(t'+1)$ for $t' < s_j(U)$ and this proves that the binary
  search upper bound of \textsc{Find-Density} is correct.

  Since the numbers involved in the arithmetic and the binary search
  are polynomial in terms of $s_j(U)$, $\cost_j(s_j(U))$, we conclude
  that the algorithms take time polynomial in the input size. 
  %
\end{proof}

 \paragraph{The online setting.}
 When the algorithm does not get to see the entire input at once, then
 it does not know $s_j(U)$. However, we can get around this by
 observing that 
 if we have sizes $t_1, \ldots, t_\ell$ satisfying $I_j(t_j) >
 \ideal_j(\rho_j) = x^*$ for all $j$, then $I_j(t_j) > x^*$ and hence
 $\min_j I_j(t_j)$ suffices as an upper bound for the binary search in
 \textsc{Find-Proper-Density}. Since we invoke \textsc{Find-Density}
 for $x < I_j(t_j)$, we have that $t_1, \ldots, t_\ell$ suffice as upper bounds for
 the binary searches in \textsc{Find-Density} as well. 

 By (A2), we have $\dcost_j(s_j(U)) > v \geq \rho_j$ for any proper
 $\rho$, so if we guess $t_j$ such that $\dcost_j(t_j) > v$, then we
 have $I_j(t_j) > \ideal_j(\rho_j)$. Therefore, we set $t_j = 2^m$,
 with $m = 1$ initially, increment $m$ one at a time and check if
 $\dcost_j(t_j) < v$. Assumption (A2) and monotonicity of
 $\dcost_j(\cdot)$ implies that $t_j \leq 2s_j(U)$ and so it takes us
 $m \leq \log(2s_j(U))$ iterations to get $\dcost_j(t_j) >
 v$. Repeating this for each dimension gives us sufficient upper
 bounds for the binary searches in both algorithms.


\section{Multi-dimensional costs with general feasibility constraint}
\label{app:multi-constrained}

In this section we consider the multi-dimensional costs setting with a general feasibility constraint, $(\profit,\Feas)$. As before we use an $O(1)$-approximate offline and an $\alpha$-competitive online algorithm for $(\Phi,\Feas)$ as a subroutine, where $\Phi$ is a sum-of-values objective. While we will be able to obtain an $O(\alpha\ell)$ approximation in the offline setting, this only translates into an $O(\alpha\ell^5)$ competitive online algorithm.

As in Section~\ref{sec:multi}, we associate with each element $e$ an $\ell$-dimensional proper density vector $\rho(e)$. Then we decompose the profit functions $\profit_i$ into sum-of-values objectives defined as follows.
\[\fdf^i_\gamma(A) = \gamma_is_i(A) - \cost_i(s_i(A)),\]
\[\sdf^i_\gamma(A) = \sum_{e \in A} (\rho_i(e) - \gamma_i)s_i(e).\]
Let $\sdf_\gamma(A) = \sum_i \sdf_\gamma^i(A)$ and $\fdf_\gamma(A) =
\sum_i \fdf_\gamma^i(A)$.  We have $\profit_i(A) = \sdf_\gamma^i(A) +
\fdf_\gamma^i(A)$ for all $i$ and $\profit(A) = \sdf_\gamma(A) +
\fdf_\gamma(A)$. As before, $\fdfmax_{\gamma, i} = \argmax_{A \subset
  \bP_\gamma} s_i(A)$.

We extend the definition of ``boundedness'' to the multi-dimensional setting as follows (see also Definition~\ref{def:bounded}).
\begin{definition}
  Given a density vector $\gamma$, a subset $A \subset P_\gamma$ is
  said to be $\gamma$-bounded if, for all $i$, $A$ is
  $\gamma_i$-bounded with respect to $\cost_i$, that is, $\rho_i(A)
  \geq \gamma_i$ and $s_i(A) \leq \dcostinv_i(\gamma_i)$. If $A$ is
  $\gamma$-bounded and $\gamma$ is the minimum density of $A$, then we
  say $A$ is \emph{bounded}.
\end{definition}


The following lemma is analogous to Proposition~\ref{lem:profit-lb}.
\begin{lemma}
  \label{lem:multi-costs:profit-lb}
  If $A$ is bounded,
  then $\profit(A) \geq \profit_i(A)$ for all $i$.
  Moreover, if $A$ is $\gamma$-bounded then for all $i$, we have
  \[\profit(A) \geq \fdf_{\gamma}(A) \geq \fdf^i_{\gamma}(A),\]
  \[\profit(A) \geq \sdf_{\gamma}(A) \geq \sdf^i_{\gamma}(A).\]
\end{lemma}

\begin{proof}
  We start by assuming that $A$ is $\gamma$-bounded.  For each
  dimension $i$, we get $\profit_i(A) \geq \fdf_\gamma^i(A) \geq 0$ by
  Proposition \ref{lem:profit-lb}. Thus, we have $\profit(A) \geq
  \fdf_\gamma(A) \geq \fdf_\gamma^i(A)$ by summing over $i$. A similar
  proof shows that $\profit(A) \geq \sdf_{\gamma}(A) \geq
  \sdf^i_{\gamma}(A)$. Next, we assume that $A$ is bounded. Let
  $\mu$ be its minimum density. Then $\sdf_\mu^i(A) \geq 0$ and
  $\fdf_\mu^i(A) \geq 0$. This gives us $\profit_i(A) = \sdf_\mu^i(A)
  + \fdf_\mu^i(A) \geq 0$ and so $\profit(A) \geq \profit_i(A)$ for
  all $i$.
\end{proof}



As in the single-dimensional setting, our approach is to find an
appropriate density $\gamma$ to bound $\sdf_\gamma(\OPT)$ and
$\fdf_\gamma(\OPT)$ in terms of the maximizers of $\sdf_\gamma^i$ and
$\fdf_\gamma^i$. We use Algorithm~\ref{Algorithm:offline-matroid}, the
offline algorithm for the single-dimensional constrained setting given
in Section~\ref{sec:constrained}, as a subroutine. We consider two
possible scenarios: either all feasible sets are bounded or there
exists an unbounded set. We use the following lemma to tackle the
first scenario.
\begin{lemma}
  \label{lem:multi-costs:all-bounded}
  Suppose all feasible sets are bounded. Let
  $D_i=\argmax_{D\in\{H_{\gamma_i}, G_{\gamma_i}, e\}_{\gamma_i}} \profit_i(D)$
  be the result of running Algorithm~\ref{Algorithm:offline-matroid}
  on the single-dimensional instance $(\profit_i, \Feas)$, where
  $\Feas$ is the underlying feasibility constraint. Then, we have
  $\sum_i \profit(D_i) \geq \frac{1}{4}\profit(\OPT)$.
\end{lemma}

\begin{proof}
  Fix a dimension $i$. Since all feasible sets are bounded, we have
  $\profit(D_i) \geq \profit_i(D_i)$ by
  Lemma~\ref{lem:multi-costs:profit-lb}. Furthermore,
  Theorem~\ref{Algorithm:offline-matroid} implies that $\profit_i(D_i)
  \geq \frac{1}{4} \profit_i(\OPT)$. Summing over all dimensions, we have
  \[\sum_i \profit(D_i) \geq \sum_i \profit_i(D_i) \geq
  \sum_i \frac{1}{4}\profit_i(\OPT) =
  \frac{1}{4}\profit(\OPT).\qedhere\]
\end{proof}

Now we handle the case when there exists an unbounded feasible set.
\begin{lemma}
  \label{lem:multi-costs:unbounded}
  Suppose there exists an unbounded feasible set. Let $\rhoc$ be the
  highest proper density such that $P_\rhoc$ contains an unbounded
  feasible set. Let $D'_i$ be the result of running
  Algorithm~\ref{Algorithm:offline-matroid} on $\bP_\rhoc$
  (i.e. ignoring elements of density at most $\rhoc$) with objective
  $\profit_i$ and feasibility constraint $\Feas$. If $e_\rhoc \in \I$,
  then we have
  \[4\sum_i \profit(D'_i) + \ell(\max_i \profit(\fdfmax_{\rhoc,i}) +
  \profit(e_\rhoc)) \geq \profit(\OPT).\]
\end{lemma}

\begin{proof}
  For brevity we use $\ebar$ to denote $e_\rhoc$. Let $\OPT_1 = \OPT
  \cap (\bP_\rhoc)$ and $\OPT_2 = \OPT \setminus \OPT_1$. By
  subadditivity, we have $\profit(\OPT) \leq \profit(\OPT_1) +
  \profit(\OPT_2)$. We observe that all feasible sets in $\bP_\rhoc$
  are bounded by definition of $\rhoc$. Thus, we can apply
  Lemma~\ref{lem:multi-costs:all-bounded} and get $4\sum_i
  \profit(D'_i) \geq \profit(\OPT_1)$.

  Next, we approximate $\profit(\OPT_2)$. We have $\sdf_\rhoc(\OPT_2)
  \leq 0$ since $\OPT_2$ has elements of density at most $\rhoc$. So,
  we have $\profit(\OPT_2) \leq \sum_i \fdf_\rhoc^i(\OPT_2)$.  We know
  that for all $i$,
  \begin{align*}
    \ideal(\ebar)
    &= \max_t (\rhoci t - \cost_i(t)) \\
    &\geq \rhoci s_i(\OPT_2) - \cost_i(s_i(\OPT_2))  \\
    &= \fdf_\rhoc^i(\OPT_2),
  \end{align*}
  so we have $\ell \ideal(\ebar) \geq \profit(\OPT_2)$ and it suffices
  to bound $\ideal(\ebar)$.


  Let $L$ be the unbounded feasible set in $P_\rhoc$; note that its
  minimum density is $\rhoc$. Since $L$ is unbounded, there exists $j$
  such that $s_j(L) > \dcostinv_j(\ebar)$. We know that
  $\fdfmax_{\rhoc,j}$ maximizes $s_j(\cdot)$ among feasible sets
  contained in $P_\rhoc \setminus \{e_\rhoc\}$. Therefore, we have
  $s_j(\fdfmax_{\rhoc,j}) + s_j(\ebar) \geq s_j(L) >
  \dcostinv_j(\ebar)$.

  Now consider two cases. Suppose that for all $i$,
  $s_i(\fdfmax_{\rhoc, i}) \leq \dcostinv_i(\ebar)$. Then,
  $s_i(\fdfmax_{\rhoc,j})\le \dcostinv_i(\ebar)$. For brevity, we
  denote $\fdfmax_{\rhoc,j}$ by $G'$. For all $i$, let $\alpha_i$ be
  such that $s_i(\fdfmax') + \alpha_is_i(\ebar) = \dcostinv_i(\ebar)$
  and let $\alpha = \min_i \alpha_i$. Without loss of generality, we
  assume that $\alpha_1$ is the minimum. By definition, we have
  $s_1(\fdfmax' + \alpha\ebar) = \dcostinv_1(\ebar)$ and $s_i(\fdfmax'
  + \alpha\ebar) \leq \dcostinv_i(\ebar)$ for $i \neq 1$. This implies
  that $\fdfmax' + \alpha\ebar$ is a $\rhoc$-bounded fractional set
  and so
  \begin{align*}
    \profit(\fdfmax' + \alpha\ebar)
    &\geq \fdf_\rhoc^1(\fdfmax' + \alpha\ebar) \\
    &= \rhoci[1] s_1(\fdfmax' + \alpha\ebar) - \cost_1(s_1(\fdfmax' +
    \alpha\ebar)) \\
    &= \rho_1(\ebar) \dcostinv_1(\ebar) - \cost_1(\dcostinv_1(\ebar)) \\
    &= \ideal(\ebar).
  \end{align*}
  Since $\ebar \in \I$, we have $\ell(\profit(\fdfmax') +
  \profit(\ebar)) \geq \ell\ideal(\ebar) \geq \profit(\OPT_2)$.


  In the second case, there exists a dimension $i$ such that
  $\dcostinv_i(\ebar) < s_i(\fdfmax_{\rhoc, i})$. Without loss of
  generality, we assume that it is the first dimension. In this case,
  we define $\fdfmax'$ to be $\fdfmax_{\rhoc,1}$. Let $\mu$ be the
  minimum density in $\fdfmax'$. 
  Note that $\mu_1 > \rhoci[1]1<$ by definition of the
  $\fdfmax_{\rhoc,i}$'s. Since $\fdfmax'$ is bounded, we have
  $\dcostinv_1(\ebar) < s_1(\fdfmax') \leq
  \dcostinv_1(\mu_1)$. Applying Lemma~\ref{lem:monotone} over $\mu_1$
  gives
  \[\fdf_{\mu}^1(\fdfmax') = \mu_1 s_1(\fdfmax') -
  \cost_1(s_1(\fdfmax')) > \mu_1 \dcostinv_1(\ebar) -
  \cost_1(\dcostinv_1(\ebar)).\]
  Since $\fdfmax'$ is $\mu$-bounded,
  by Lemma~\ref{lem:multi-costs:all-bounded} we have
  \begin{align*}
    \profit(\fdfmax') 
    &\geq \fdf_\mu^1(\fdfmax') \\
    &> \mu_1 \dcostinv_1(\ebar) - \cost_1(\dcostinv_1(\ebar))\\
    &> \rhoci[1] \dcostinv_1(\ebar) -
    \cost_1(\dcostinv_1(\ebar))\\
    &= \ideal(\ebar).
  \end{align*}
  This gives us $\ell\profit(\fdfmax') = \profit(\OPT_2)$. Overall, we
  get $\ell(\max_i\profit(\fdfmax_{\rhoc, i}) + \profit(\ebar)) \geq
  \profit(\OPT_2)$.
\end{proof}

We are now ready to give an offline algorithm. Let $D_{\gamma, i}$ be
the result of running Algorithm~\ref{Algorithm:offline-matroid} on
$P_\gamma$ with objective $\profit_i$ and feasibility constraint
$\Feas$. 
\begin{algorithm}
  \caption{Offline algorithm for multi-dimensional $(\profit, \Feas)$}
\label{Algorithm:multi-costs:offline-matroid}
\begin{algorithmic}[1]
  \STATE Let $G_{\max} = \argmax_{G\in\{G_{\gamma,i}\}_{\gamma,i}} \profit(G)$.
  \STATE Let $D_{\max} = \argmax_{D\in\{D_{\gamma,i}\}_{\gamma,i}} \profit(D)$.
  \STATE Let $e_{\max} = \argmax_{e\in U} \profit(e)$.
  \RETURN the most profitable of $G_{\max{}}$, $D_{\max{}}$, and $e_{\max{}}$.
\end{algorithmic}
\end{algorithm}


Finally, we use the above lemmas to lower bound the performance of
Algorithm \ref{Algorithm:multi-costs:offline-matroid}.
\begin{theorem}
  \label{thm:multi-costs:offline-matroid}
 Algorithm \ref{Algorithm:multi-costs:offline-matroid} gives a
 $7\ell$-approximation to $(\profit,\Feas)$ where $\profit$ is an $\ell$-dimensional profit function and $\Feas$ is a matroid feasibility constraint.
\end{theorem}
\begin{proof}
  First, we consider only elements from $\I$.  If all feasible sets are
  bounded, then we get $\profit(D_{\max{}}) \geq \frac{1}{4\ell}\profit(\OPT)$
  by Lemma~\ref{lem:multi-costs:all-bounded}. On the other hand, if
  there exists an unbounded feasible set, then we have
  \begin{align*}    
    \ell(4\profit(D_{\max{}}) + \profit(e_{\max{}}) +
    \profit(\fdfmax_{\max{}}))
    &\geq 4\sum_i \profit(D'_i) + \ell(\profit(\ebar) + 
    \max_i \profit(\fdfmax_{\rhoc,i}))\\
    &\geq \profit(\OPT \cap \I)
  \end{align*}
  by Lemma~\ref{lem:multi-costs:unbounded}. For elements from $\F$,
  Lemma~\ref{lem:multi-costs:F} shows that $\profit(e_{\max{}}) \geq
  \profit(\OPT \cap \F)/\ell$. This proves that the algorithm achieves a
  $7\ell$-approximation.
\end{proof}

\subsection{The online setting}

\newcommand{\emax}{e_{\max{}}}
\newcommand{\Omax}{O_{0}}
\newcommand{\bigC}{\lambda}
\newcommand{\smallC}{c}

\begin{algorithm}
  \caption{Online algorithm for multi-dimensional $(\profit, \Feas)$}
\label{Algorithm:multi-costs:online-matroid}
\begin{algorithmic}[1]
\STATE Let $c$ be a uniformly random draw from $\{1,2,3\}$.
\STATE Let $i^*$ be a uniformly random draw from $\{1, \ldots, \ell\}$.
\IF {$c=1$}
  \RETURN $\Omax$: the result of running Dynkin's online algorithm.
\ELSIF {$c=2$}
  \RETURN $O_1$: the result of running Algorithm~\ref{Algorithm:online-matroid} on
  $(\profit_{i^*}, \Feas)$.
\ELSIF {$c=3$}
  \STATE Draw $k$ from $\textrm{Binomial}(n,1/2)$.
  \STATE Let the sample $X$ be the first $k$ elements and $Y$ be the remaining elements.
 \STATE Determine $\alg(X)$ using the offline Algorithm \ref{Algorithm:multi-costs:offline-matroid}.
 \STATE Let $\eone$ be a specified constant and let $\guess$ be the highest density such that
 $\profit(\alg(\prefguess^X)) \geq \frac{\eone}{49\ell^2}\profit(\alg(X))$.
 \RETURN $O_2$: the result of running
 Algorithm~\ref{Algorithm:online-matroid} on $\prefguess^Y$ 
 with objective $\profit_{i^*}$ and feasibility constraint $\Feas$.

\ENDIF 
\end{algorithmic}
\end{algorithm}

In this subsection, we develop the online algorithm for constrained multidimensional profit maximization. First we remark that density prefixes are well-defined in this setting, so we can use
Algorithm~\ref{Algorithm:online-matroid}, the online algorithm for single-dimensional constrained profit maximization, in steps 6 and 12. In order to mimic the offline algorithm, we guess (with equal probability) among one of the following scenarios and apply the appropriate subroutine for profit maximization: if there is a single high profit element we apply Dynkin's algorithm, if all feasible sets are bounded we apply Algorithm \ref{Algorithm:online-matroid} along a randomly selected dimension, else if there exists an unbounded feasible set we first estimate a density threshold by sampling and then apply Algorithm \ref{Algorithm:online-matroid} over a random dimension but restricted to elements with density above the threshold. 
   
Most of this subsection is devoted the analysis of
Algorithm~\ref{Algorithm:multi-costs:online-matroid} when there exists
an unbounded feasible set. 
We prove the overall competitive ratio of the algorithm in Theorem \ref{theorem:multi-costs-constrained-online}.

We define $\rhob$ as follows.
\begin{definition}
  For a fixed parameter $\eone \leq 1 $, let $\rhob$ be the highest
  density such that $\profit(\OPT(\prefsmall)) \geq
  \left(\frac{\eone}{49\ell^2}\right)^2 \profit(\OPT)$.
\end{definition}

The following lemma is essentially an analogue of Lemma \ref{lem:main}.
\begin{lemma}
  \label{lem:multi-costs:main}
  Suppose there exists an unbounded feasible set.
  For fixed parameters $\const[1] \geq 1 $, $\const[4] \leq 1$, $\eone
  \leq 1$ and $\etwo \leq 1$ assume that there does not exists an element with
  profit more than $\frac{1}{\const[1]}\profit(\OPT)$. Then with
  probability at least $\const[4]$, we have that $\guess$, as defined
  in Algorithm~\ref{Algorithm:multi-costs:online-matroid}, satisfies
  \begin{enumerate}
  \item $ \rhob \geq \guess \geq \rhoc$ and 
  \item $\profit(\OPT(\prefguess^Y)) \geq \etwo
    \left(\frac{\eone}{49\ell^2}\right)^2 \profit(\OPT)$
  \end{enumerate}
\end{lemma}

\begin{proofsketch}
  We define events $E_1$ and $E_2$ as in the
  single-dimensional setting:
\begin{align*}
E_1: \profit(\OPT(\preflarge^X)) & \geq \eone \profit(\OPT(\preflarge)) \\ 
E_2: \profit(\OPT(\prefsmall^Y)) & \geq \etwo \profit(\OPT(\prefsmall)).
\end{align*}
We observe that the proof of Lemma \ref{lem:main} primarily depends on
the properties of density prefixes, in particular that
$\profit(\OPT(P^X_\rhoc))$ is a sufficiently large fraction of
$\profit(\OPT(\preflarge))$, and that
Algorithm~\ref{Algorithm:offline-matroid}'s profit does not decrease
when considering larger
prefixes. 
Then Theorem~\ref{thm:multi-costs:offline-matroid} gives us
$\profit(\OPT(P_\rhoc)) \geq \frac{1}{7\ell} \profit(\OPT)$.
  
Thus, a proof similar to that of Lemma~\ref{lem:main} validates
Lemma~\ref{lem:multi-costs:main}.
\end{proofsketch}


The following lemma establishes the competitive ratio of Algorithm \ref{Algorithm:multi-costs:online-matroid} in the presence of an unbounded feasible set. 

\begin{lemma}
  \label{lem:multi-costs:secondstage}  
  Suppose that there exists an unbounded feasible set, and let
  $\alpha$ be the competitive ratio that
  Algorithm~\ref{Algorithm:online-matroid} achieves for a
  single-dimensional problem $(\profit_i,\Feas)$. Then, for a fixed
  set $Y$ and threshold $\guess$, satisfying $\guess \geq \rhoc$, we
  have $ \E_\sigma[\profit(O_2)] \geq \frac{1}{\alpha\ell}
  \ \profit(\OPT(\prefguess^Y))$, where the expectation is over all
  permutations $\sigma$ of $Y$.
\end{lemma}

\begin{proof}
  First, we prove that $O_2$ is bounded. By the definition of
  $\rhoc$, since $O_2 \subseteq \prefguess^Y$ it suffices to show
  that $\prefguess^Y \subseteq \bP_\rhoc$, as all feasible sets in $\bP_\rhoc$ are bounded. The
  threshold $\tau$ is either equal to or strictly greater than $\rhoc$
  and it suffices to consider the former case. If $\tau = \rhoc$, then
  $e_\rhoc$ must have been in the sample set $X$ so it is not in
  $\prefguess^Y$ and so $\prefguess^Y \subseteq \bP_\rhoc$.
 
  Hence, by Lemma~\ref{lem:multi-costs:profit-lb}, for all $i$ we have
  $\profit(O_2) \geq \profit_i(O_2)$. Applying
  Theorem~\ref{thm:online-matroid} where the ground set is
  $\prefguess^Y$ gives us $\E_\sigma[ \profit_i(O_2)\ |\ i^*=i] \geq
  \frac 1\alpha \profit_i(\OPT(\prefguess^Y))$. 
  Therefore, we have
  \begin{align*}
    \E_{\sigma}[\profit(O_2)] 
    &= \frac{1}{\ell}\sum_i \E_\sigma[\profit(O_2)\
    |\ i^* = i]\\
    &\geq \frac{1}{\ell}\sum_i \E_\sigma[\profit_i(O_2)\
    |\ i^* = i]\\
    &\geq \frac{1}{\alpha\ell}\sum_i \profit_i(\OPT(\prefguess^Y))\\
    &= \frac{1}{\alpha\ell} \profit(\OPT(\prefguess^Y)).\qedhere
  \end{align*}
\end{proof}

We now prove the main result of this section. 
\begin{theorem}
 \label{theorem:multi-costs-constrained-online} 
 Let $\alpha$ denote the competitive ratio of
 Algorithm~\ref{Algorithm:online-matroid} for the single-dimensional
 problem $(\profit_i,\Feas)$. Then
 Algorithm~\ref{Algorithm:multi-costs:online-matroid} achieves a
 competitive ratio of $O(\alpha\ell^5)$ for the multi-dimensional
 problem $(\profit, \Feas)$.
\end{theorem}
\begin{proof}
  If all feasible sets are bounded, we can apply
  Lemma~\ref{lem:multi-costs:profit-lb}. Then,
  \begin{align*}
    \E[\profit(O_1)] 
    &\geq \frac{1}{\ell}\sum_i\E[\profit_i(O_1)\ |\ i^* = i]\\
    &\geq \frac{1}{\alpha\ell}\sum_i\profit_i(\OPT)\\
    &= \frac{1}{\alpha\ell} \profit(\OPT).
  \end{align*}

  Suppose there exists an element with profit
 more than $\frac{1}{\const[1]} \profit(\OPT)$. Since we apply the
  standard secretary algorithm with probability $\frac{1}{3}$, and the
  standard secretary algorithm is $e$-competitive, in expectation
  we get a profit of $\frac{1}{3 \const[1] e}$ times the optimal.

  Finally, if no element has profit more than $\frac{1}{\const[1]}
  \profit(\OPT)$ and there exists an unbounded feasible set, using the
  second inequality of Lemma \ref{lem:multi-costs:main}, Lemma
  \ref{lem:multi-costs:secondstage}, and the fact that we output $O_2$
  with probability $\frac{1}{3}$, we get
  \begin{align*}
    \E[\profit(O_2)]
    &\geq \frac{1}{3}\E[\profit(O_2)\ |\ E_1 \wedge
    E_2]\Pr[E_1\wedge E_2]\\
   &\geq \frac{\const[4] }{3\alpha\ell}\cdot\E[\profit(\OPT(\prefguess^Y))\
    |\ E_1 \wedge E_2]\\
    &\geq \frac{\const[4] \etwo}{3\alpha\ell}\left(\frac{\eone}{49\ell^2}\right)^2
    \profit(\OPT).
  \end{align*}
  Overall, we have that \[\E[\profit(O)] \geq \min\left\{\frac{
      \const[4]\etwo}{3\alpha\ell}\left(\frac{\eone}{49\ell^2}\right)^2, \frac{1}{3
      \const[1] e}, \frac{1}{3\alpha\ell}\right\} \cdot \profit(\OPT).\] Since all the
  involved parameters are fixed constants we get the desired result.
\end{proof}

\bibliographystyle{plain}
\bibliography{matroid_secretary}

\begin{thebibliography}{10}

\bibitem{babaioff2009secretary}
M.~Babaioff, M.~Dinitz, A.~Gupta, N.~Immorlica, and K.~Talwar.
\newblock {Secretary problems: weights and discounts}.
\newblock In {\em Proceedings of the twentieth Annual ACM-SIAM Symposium on
  Discrete Algorithms}, pages 1245--1254. Society for Industrial and Applied
  Mathematics, 2009.

\bibitem{babaioff2008selling}
M.~Babaioff, J.~Hartline, and R.~Kleinberg.
\newblock {Selling banner ads: Online algorithms with buyback}.
\newblock In {\em Fourth Workshop on Ad Auctions}, 2008.

\bibitem{babaioff2009selling}
M.~Babaioff, J.~D. Hartline, and R.~D. Kleinberg.
\newblock {Selling ad campaigns: Online algorithms with cancellations}.
\newblock In {\em Proceedings of the tenth ACM conference on Electronic
  commerce}, pages 61--70. ACM, 2009.

\bibitem{babaioff2007knapsack}
M.~Babaioff, N.~Immorlica, D.~Kempe, and R.~Kleinberg.
\newblock A knapsack secretary problem with applications.
\newblock In {\em Approximation, Randomization, and Combinatorial Optimization.
  Algorithms and Techniques}, pages 16--28. Springer Berlin / Heidelberg, 2007.

\bibitem{babaioff2007matroids}
M.~Babaioff, N.~Immorlica, and R.~Kleinberg.
\newblock {Matroids, secretary problems, and online mechanisms}.
\newblock In {\em Proceedings of the eighteenth annual ACM-SIAM symposium on
  Discrete algorithms}, pages 434--443. Society for Industrial and Applied
  Mathematics, 2007.

\bibitem{bateni2010submodular}
M.~H. Bateni, M.~T. Hajiaghayi, and M.~Zadimoghaddam.
\newblock {Submodular secretary problem and extensions}.
\newblock In {\em Approximation, Randomization, and Combinatorial Optimization.
  Algorithms and Techniques}, pages 39--52. Springer Berlin / Heidelberg, 2010.

\bibitem{chakrabortyimproved}
S.~Chakraborty and O.~Lachish.
\newblock Improved competitive ratio for the matroid secretary problem.
\newblock To appear in {\it SODA '12}.

\bibitem{CHMS10}
S.~Chawla, J.~D. Hartline, D.~L. Malec, and B.~Sivan.
\newblock Multi-parameter mechanism design and sequential posted pricing.
\newblock In {\em Proceedings of the 42nd ACM symposium on Theory of
  computing}, STOC '10, pages 311--320. ACM, 2010.

\bibitem{dean2005pip}
B.~C. Dean, M.~X. Goemans, and J.~Vondr\'{a}k.
\newblock Adaptivity and approximation for stochastic packing problems.
\newblock In {\em Proceedings of the sixteenth annual ACM-SIAM symposium on
  Discrete algorithms}, SODA '05, pages 395--404. Society for Industrial and
  Applied Mathematics, 2005.

\bibitem{dimitrov2008}
N.~B. Dimitrov and C.~G. Plaxton.
\newblock Competitive weighted matching in transversal matroids.
\newblock In {\em Proceedings of the 35th international colloquium on Automata,
  Languages and Programming, Part I}, pages 397--308. Springer-Verlag, 2008.

\bibitem{Dynkin1963}
E.~B. Dynkin.
\newblock {The optimum choice of the instant for stopping a Markov process}.
\newblock {\em Soviet Math. Dokl}, 4(627-629), 1963.

\bibitem{Feige2005}
U.~Feige, A.~Flaxman, J.~D. Hartline, and R.~Kleinberg.
\newblock {On the Competitive Ratio of the Random Sampling Auction}.
\newblock In {\em Proceedings of the first annual workshop on Internet and
  network economics}. Springer, 2005.

\bibitem{FNS11}
Moran Feldman, Joseph~Seffi Naor, and Roy Schwartz.
\newblock Improved competitive ratios for submodular secretary problems.
\newblock In {\em Proceedings of the 14th international workshop and 15th
  international conference on Approximation, randomization, and combinatorial
  optimization: algorithms and techniques}, APPROX'11/RANDOM'11, pages
  218--229, Berlin, Heidelberg, 2011. Springer-Verlag.

\bibitem{Fer89}
T.~Ferguson.
\newblock Who solved the secretary problem.
\newblock {\em Statist. Sci.}, 4(3):282--289, 1989.

\bibitem{Fre83}
P.~R. Freeman.
\newblock The secretary problem and its extensions: a review.
\newblock {\em International Statistical Review}, 51(2):189--206, 1983.

\bibitem{gupta2010constrained}
Anupam Gupta, Aaron Roth, Grant Schoenebeck, and Kunal Talwar.
\newblock Constrained non-monotone submodular maximization: Offline and
  secretary algorithms.
\newblock In {\em Internet and Network Economics}, pages 246--257. Springer
  Berlin / Heidelberg, 2010.

\bibitem{HKP04}
Mohammad~Taghi Hajiaghayi, Robert Kleinberg, and David~C. Parkes.
\newblock Adaptive limited-supply online auctions.
\newblock In {\em Proceedings of the 5th ACM conference on Electronic
  commerce}, EC '04, pages 71--80, New York, NY, USA, 2004. ACM.

\bibitem{kennedy1987prophet}
D.~P. Kennedy.
\newblock {Prophet-type inequalities for multi-choice optimal stopping}.
\newblock {\em Stochastic Processes and their Applications}, 24(1):77--88,
  1987.

\bibitem{Kleinberg2005multiple}
R.~Kleinberg.
\newblock {A multiple-choice secretary algorithm with applications to online
  auctions}.
\newblock In {\em Proceedings of the sixteenth annual ACM-SIAM symposium on
  Discrete algorithms}, pages 630--631. Society for Industrial and Applied
  Mathematics, 2005.

\bibitem{korulapal2009}
N.~Korula and M.~P\'{a}l.
\newblock Algorithms for secretary problems on graphs and hypergraphs.
\newblock In {\em Proceedings of the 36th International Colloquium on Automata,
  Languages and Programming: Part II}, ICALP '09, pages 508--520.
  Springer-Verlag, 2009.

\bibitem{Oxley1992}
J.~Oxley.
\newblock {\em Matroid Theory}.
\newblock Oxford University Press, 1992.

\bibitem{samuel1984comparison}
E.~Samuel-Cahn.
\newblock {Comparison of threshold stop rules and maximum for independent
  nonnegative random variables}.
\newblock {\em The Annals of Probability}, 12(4):1213--1216, 1984.

\bibitem{Sam91}
S.~Samuels.
\newblock Secretary problems.
\newblock In {\em Handbook of Sequential Analysis}, pages 381--405. Marcel
  Dekker, 1991.

\bibitem{soto2010matroid}
J.~A. Soto.
\newblock {Matroid Secretary Problem in the Random Assignment Model}.
\newblock In {\em Proceedings of the twenty-second Annual ACM-SIAM Symposium on
  Discrete Algorithms}, pages 1275--1284. Society for Industrial and Applied
  Mathematics, 2011.

\end{thebibliography}

 
\end{document}